\documentclass[11pt]{article}

\usepackage{fullpage}
\usepackage{natbib}

\usepackage{times}
\usepackage{setspace}

\usepackage{amsmath, amssymb, amsthm, dsfont}

\usepackage{thm-restate}
\usepackage{xspace}

\usepackage{amsmath,amsfonts,bm}

\def\eqref#1{equation~\ref{#1}}

\def\1{\bm{1}}

\def\eps{{\varepsilon}}
\providecommand{\email}[1]{\href{mailto:#1}{\nolinkurl{#1}\xspace}}

\DeclareMathAlphabet{\mathsfit}{\encodingdefault}{\sfdefault}{m}{sl}
\SetMathAlphabet{\mathsfit}{bold}{\encodingdefault}{\sfdefault}{bx}{n}

\usepackage{hyperref}
\usepackage{url}
\hypersetup{
    colorlinks=true,
    linkcolor=blue,
    urlcolor=cyan,
    citecolor=blue,
}

\newtheorem{theorem}{Theorem}
\newtheorem{lemma}{Lemma}
\newtheorem{proposition}{Proposition}
\newtheorem{corollary}{Corollary}

\theoremstyle{definition}
\newtheorem{definition}{Definition}

\newcommand{\delete}[1]{}

\title{High Probability Streaming Lower Bounds for $F_2$ Estimation}

\author{
William Swartworth\thanks{The Voleon Group. 
E-mail: \email{wswartwo@andrew.cmu.edu}.
Work done in part while at Carnegie Mellon University. 
}
\and
David P. Woodruff\thanks{Carnegie Mellon University and Google Research. 
E-mail: \email{dwoodruf@cs.cmu.edu}. 
David P. Woodruff is supported in part by Office of Naval Research award number N000142112647 and a Simons Investigator Award.}
\and
Samson Zhou\thanks{Texas A\&M University. 
E-mail: \email{samsonzhou@gmail.com}. 
Supported in part by NSF CCF-2335411. 
The author gratefully acknowledges funding provided by the Oak Ridge Associated Universities (ORAU) Ralph E. Powe Junior Faculty Enhancement Award.}
}

\date{}  

\begin{document}

\maketitle

\begin{abstract}
Estimating the second frequency moment ($F_2$) of an underlying frequency vector is a fundamental problem in the streaming model. While recent work by Braverman and Zamir [STOC 2025] resolved the space complexity for constant failure probability in the insertion-only model, the optimal dependence on the failure parameter $\delta$ remained open.

We close this gap by proving a tight high-probability lower bound of $\Omega\left(\frac{1}{\eps^2}\log\frac{1}{\delta}\,\log\frac{\eps\sqrt{n}}{\log(1/\delta)}\right)$
for $(1\pm\eps)$-approximate $F_2$ estimation. The key challenge is the failure of prior multi-scale direct sum arguments under noise sensitivity. We introduce a noise-robust communication primitive, \textit{Exam Mostly Set Disjointness}, and prove an $\Omega\left(\frac{m}{t}\log\frac{1}{\delta}\right)$ one-way lower bound. Embedding this into a multi-scale reduction yields the correct $\log(1/\delta)$ dependence.

We also give two complementary algorithms under natural structure assumptions. For streams with frequency bound $B$, we design a subsampling method using continuous $F_0$ tracking that replaces a $\log(n)$ factor with $\text{polylog}(B)$. For $k$-sparse streams, we develop a two-stage sketch using approximate Morris counters, replacing $\log n$ with $\log k$ and achieving a further $\log\log m$ dependence on stream length.
\end{abstract}

\section{Introduction}
The emergence of large-scale data applications has made streaming algorithms a fundamental framework for space-efficient computation.
Such applications routinely generate massive datasets arising from sources such as social networks, web search systems, network monitoring, and distributed sensor platforms. 
A central problem in this model is the estimation of frequency moments. 
Here, elements from a universe $[n]:=\{1,2,\ldots,n\}$ arrive one by one in a stream, implicitly defining a frequency vector $f$, where $f_i \in \mathbb{N}$ denotes the number of occurrences of each element $i \in[n]$. 
We typically denote the stream length by $m$ and assume that $n = \mathrm{poly}(m)$. 
The goal is to output, with high probability, a $(1 \pm \eps)$-approximation to the $p$-th frequency moment of the stream, $F_p := \sum_{x \in U} f_x^p$, using space sublinear in both $m$ and $n$. 
Since the seminal work of Alon, Matias, and Szegedy~\cite{AlonMS99}, the frequency moments problem has been closely studied for more than two decades~\cite{ChakrabartiKS03,Bar-YossefJKS04,Woodruff04,IndykW05,Indyk06,Li08,KaneNW10,KaneNPW11,Ganguly11,BravermanO13,BravermanKSV14,BlasiokDN17,BravermanVWY18,GangulyW18,WoodruffZ21b,WoodruffZ21,Ben-EliezerEO22,JayaramWZ24,WoodruffZ24,GribelyukLWYZ25,GribelyukLWYZ26,Lin0W0Z26}.

The second frequency moment ($p=2$) is a fundamental quantity in data stream analysis. 
This statistic, commonly referred to as the repeat rate or surprise index, appears in a range of settings including database query optimization \cite{AlonGMS02}, approximate histogram maintenance \cite{GilbertGIKMS02}, machine learning \cite{muthukrishnan2005data, woodruff2014sketching}, and network traffic anomaly detection \cite{KrishnamurthySZC03}, among others. 
In addition to these applications, $F_2$ provides a measure of skew or homogeneity in a dataset, which is important for performance optimization in database systems. 
In distributed and parallel computation, it is used to assess load imbalance and to guide how data should be partitioned across processors. 
It also informs other streaming tasks; for instance, some hybrid algorithms for estimating the number of distinct elements ($F_0$) rely on $F_2$ to decide which sampling strategy to apply.

Exactly computing $F_2$ in a streaming setting requires storing all frequency counts, which leads to $\Omega(n)$ space usage and becomes impractical for large inputs. 
This has driven the development of streaming algorithms that maintain compact summaries of the data and produce accurate estimates in a single pass.
Formally, the goal of the $F_2$ estimation problem is to compute a $(1 \pm \eps)$-approximation with probability at least $1 - \delta$ while minimizing space. 
The classical AMS sketch of Alon, Matias, and Szegedy~\cite{AlonMS99} achieves this using $O\left(\frac{1}{\eps^2}\log\frac{1}{\delta}\log n\right)$ bits of space for a universe of size $n$. 
Despite extensive progress, the best known lower bound until 2025, due to Woodruff~\cite{Woodruff04}, was $\Omega\left(\frac{1}{\eps^2} + \log n\right)$. 
Finally, a recent breakthrough by Braverman and Zamir \cite{BravermanZ25} closed this gap for the case of constant failure probability $\delta$ in the insertion-only model, proving a tight lower bound of $\Omega(\frac{1}{\eps^2} \log(n\eps^2))$ space. 
When $\eps$ is not too small, e.g., $\eps > n^{-1/2+c}$ for $c>0$, this matches the AMS upper bound of $O\left(\frac{1}{\eps^2} \log n\right)$. 
For smaller $\eps$, \cite{BravermanZ25} also provided a modified algorithm matching their lower bound. 
However, while the complexity in terms of $n$ and $\eps$ was resolved, the optimal dependence on the failure probability $\delta$ remains an open question. 

Our main result closes this gap by establishing a matching lower bound on the space complexity of $F_2$ estimation as a function of the failure probability $\delta$, thereby resolving the dependence on $\delta$ in the insertion-only streaming model:

\begin{restatable}{theorem}{thmmain}
\label{thm:main}
Let $\mathcal{A}$ be a streaming algorithm that, for any data stream of length polynomial in $n$ over a universe of size $n$, computes an estimate $\hat{F_2}$ such that $\Pr[|\hat{F_2} - F_2| \le \eps F_2] \ge 1 - \delta$. For $\eps\sqrt{n}\geq \log\frac{1}{\delta}$,
the space used by $\mathcal{A}$ is at least:
$$ \Omega\left(\log\left(\frac{\eps\sqrt{n}}{\log\frac{1}{\delta}}\right)\frac{1}{\eps^2} \log\left(\frac{1}{\delta}\right)\right).$$
\end{restatable}

Beyond the main lower bound, our techniques yield additional results at both the communication and algorithmic levels. 
At the core of our approach is a new communication primitive called Exam Mostly Set Disjointness (EMostlyDISJ). 
In this problem, $t$ players hold sets with a total size parameter $m$, and a referee holds a target element. 
The players' sets are promised to be either $\kappa$-almost disjoint (meaning the total number of repeated occurrences of elements across the sets is at most $\kappa$) or to have a unique element common to at least a constant fraction of the sets. 

The referee must decide whether the latter case holds and the common element matches their target.

\begin{theorem}[EMostlyDISJ lower bound]
Any one-way communication protocol for the $t$-player EMostlyDISJ problem with failure probability at most $\delta$ requires $\Omega\left(\frac{m}{t}\log\frac{1}{\delta}\right)$ communication when the almost-disjointness parameter is $\kappa = O\left(\log\frac{1}{\delta}\right)$. 
(Informal, see Corollary~\ref{cor:comm_cost_final}.)
\end{theorem}
\noindent
We then use this lower bound as a core primitive in our reduction framework, together with a multi-scale embedding argument, to obtain the streaming lower bound above.

On the algorithmic side, we identify structured regimes where faster algorithms are possible.

\begin{theorem}[Bounded-frequency streams]
Let $x \in \mathbb{Z}_{\ge 0}^n$ be the frequency vector of an insertion-only stream of length $m$. 
Suppose that $0 \leq x_i \leq B$ for all $i \in [n]$. 
Then there exists a streaming algorithm that, with probability at least $1 - \delta$, computes a $(1 \pm \eps)$-approximation to $\|x\|_2^2$ using space (in bits)
\[O\left(\frac{1}{\eps^2}\log^2\frac{B}{\eps}\log\frac{1}{\delta}\left(\log B + \log\frac{1}{\eps}\right)\right) + O\left(\log\frac{n}{\delta}\right).\]
\end{theorem}

\begin{theorem}[Sparse streams]
Let $x \in \mathbb{Z}_{\geq 0}^n$ be the frequency vector of an insertion-only stream of length $m$. Suppose that $\|x\|_0 \leq k$. 
Then there exists a streaming algorithm that, with probability at least $1 - \delta$, computes a $(1 \pm \eps)$-approximation to $\|x\|_2^2$ using space (in bits)
\[O\left(\frac{1}{\eps^2}\log\left(\frac{1}{\delta}\right)\left(\log\frac{k}{\eps} + \log\log m\right) + \log n \cdot \log\frac{1}{\delta} \right).\]
\end{theorem}

Together, these results provide a refined picture of the complexity of $F_2$ estimation, capturing both its limitations in the general case and improved guarantees under natural structural assumptions.

\section{Background and Technical Overview}

Proving strong streaming space lower bounds via communication complexity often encounters a known barrier: traditional reductions struggle to amplify the hardness of a single communication problem to achieve the logarithmic factors (e.g., $\log n$) required for $F_2$ estimation. 
A recent breakthrough by Braverman and Zamir \cite{BravermanZ25} overcame this barrier by introducing a novel multi-scale direct sum argument applied directly at the streaming level, built on top of a specialized communication model.

\paragraph{The Braverman-Zamir Approach and Exam Set Disjointness.} 
At the core of their lower bound is the \emph{Exam Set Disjointness} (EDISJ) problem. In this communication game, $t$ players each hold a set of items, and a referee holds a single target element $y$. The players communicate sequentially, eventually sending a message to the referee, who must determine if $y$ is the \emph{unique} intersection of all the players' sets. This ``exam'' model is crucial because it forces the players to communicate enough information to verify a specific intersection, rather than merely detecting whether any intersection exists. They showed that solving EDISJ requires $\Omega(m/t)$ communication, where $m$ is the total size of the sets.

To lift this communication bound to the streaming setting, Braverman and Zamir embed multiple instances of EDISJ into a single stream. These instances encode games with varying numbers of players (e.g., $2, 4, 8, \ldots$, up to roughly $\eps\sqrt{n}$) and are interleaved at different ``scales'' or ``levels.'' Because these instances share parts of the stream, they are dependent, and standard direct sum theorems do not apply. Instead, they developed a clever information-theoretic argument to show that a streaming algorithm must effectively solve all these dependent instances simultaneously. At any arbitrary point $j$ in the stream, an algorithm needs information about the ``past'' relevant to each scale. For a level $l$ (corresponding to $2^l$ players), the algorithm must remember information about the preceding $n/2^l$ elements of the stream to solve the EDISJ instance at that scale. They demonstrated that these different-scaled pasts are sufficiently disjoint that the required information adds up. Summing the $\Omega(1/\eps^2)$ cost across each of the $\Theta(\log(n\eps^2))$ scales yields their final $\Omega(\frac{1}{\eps^2} \log n)$ space bound.

\paragraph{The Barrier to Optimal Error Dependence.} 
While this multi-scale argument elegantly resolves the space complexity in terms of $n$ and $\eps$, it fails to capture the optimal dependence on the failure probability $\delta$. The limitation stems directly from the EDISJ problem itself: the standard Set Disjointness problem, and its exam variant, are highly sensitive to noise. Because the YES instance relies on a single, perfect intersection, an algorithm that is permitted to fail with probability $\delta$ might simply miss this unique intersection without paying the full information cost. Consequently, EDISJ cannot yield the robust $\log(1/\delta)$ communication lower bounds necessary for stronger high-probability streaming guarantees.

\paragraph{Our Contributions: Exam Mostly Set Disjointness.} 
To achieve tight dependence on $\delta$, we adapt the Braverman-Zamir framework by replacing EDISJ with a new, noise-robust communication primitive that we call \emph{Exam Mostly Set Disjointness} (EMostlyDISJ). This problem is a synthesis of the Exam Set Disjointness problem from \cite{BravermanZ25} and the Mostly Set Disjointness problem introduced by Kamath et al. \cite{KPW21}. 

Instead of demanding perfect disjointness, EMostlyDISJ relaxes the conditions to make the problem strictly harder for algorithms prone to error. By distinguishing between sets being almost disjoint (few intersections) and having an element common to a large fraction of the sets, the problem becomes robust to failure probabilities. 

\begin{definition}[Exam Mostly Set Disjointness (EMostlyDISJ)]
The setup involves $t$ players and one referee. Let $U$ be a universe of size $|U|$. The inputs are sets $S_1, \ldots, S_t \subseteq U$ for the players and an element $j \in U$ for the referee. 
The input sets are promised to be either (i) $\kappa$-almost disjoint (the total number of repeated occurrences of elements across the sets is at most $\kappa$) or (ii) to have a unique element $j_0 \in U$ that is common to at least $ct$ of the sets for some constant $c \in (0,1)$, and the remaining elements across the sets are $\kappa$-almost-disjoint. 
The communication is one-way from player $i$ to player $i+1$, and finally to the referee. 
With failure probability at most $\delta$, the referee must decide if the input is an instance of case (ii) and the intersecting element $j_0$ is equal to its element $j$.
\end{definition}

We prove that this game requires $\Omega\left(\frac{m}{t} \log \frac{1}{\delta}\right)$ communication when the almost-disjointness parameter is bounded by $M = O\left(\log \frac{1}{\delta}\right)$.

\paragraph{Putting it Together: The Multi-Scale Reduction.} 
With the robust communication bound in hand, we plug EMostlyDISJ directly into the multi-scale direct sum framework. In our reduction, the stream is divided into blocks corresponding to players, and elements are grouped into ``super-items'' of size $d = \Theta(\frac{\eps^2 n}{t^2})$. If a super-item is repeated among a constant fraction of the players (the ``Yes'' case of EMostlyDISJ), the $F_2$ moment of the stream increases significantly compared to the ``No'' case, where items are repeated at most $M$ times. 
By embedding these robust EMostlyDISJ instances across multiple scales (different values of $t$) and applying the direct sum argument over the disjoint ``pasts,'' the robust $\log(1/\delta)$ communication costs naturally accumulate at every valid level of the stream, culminating in our complete and tight space lower bound for $F_2$ estimation.

\paragraph{Bounded-frequency streams.}
We consider two classes of insertion-only streams for which the general space lower bound above does not apply. 
By exploiting structural assumptions about the data, we show that it is possible to beat the lower bound with more space-efficient algorithms. 
The first class is streams with entries bounded by $B$, i.e., the maximum frequency of any element is $x_i \leq B$. 
The main observation is that when $B$ is small, we obtain a good estimate of $F_2(x)$ by subsampling entries from the support of $x$. 
We implement this by applying a pairwise independent hash function to filter the stream down to an expected effective dimension of roughly $K = O(\frac{B^2}{\eps^2})$. Because the frequencies are strictly bounded by $B$, the variance of this sampling-based estimator remains well-controlled. 

This subsampling reduces the size of the vector dramatically. 
When we apply the classical AMS sketch of \cite{AlonMS99} purely to this restricted substream, the counters accumulate much smaller values (bounded by $O(BK)$). 
This drastically lowers the bit complexity per entry, effectively reducing the dominant multiplicative space dependence on $\log n$ down to roughly $\log B$. 
A full statement is given in Theorem \ref{thm:l2_bounded_freq_improved}.

There are several complications in following the outlined procedure in a single-pass streaming model. In particular, the optimal subsampling rate needs to depend inversely on the total support size ($F_0$), which we do not have access to at the start of the stream and grows dynamically. To handle this, we use a continuous $F_0$ tracker to estimate the sparsity at all points in time. As the stream progresses and our $F_0$ estimate crosses power-of-two thresholds, we dynamically spawn new sets of sketches tuned to progressively smaller sampling rates. By maintaining only a small active window of parallel instances (specifically, $O(\log(B/\eps))$ instances) at any given time and discarding older ones, we successfully isolate a sketch with the correct sampling scale for the suffix of the stream, while bounding the error accumulated from the prefix.

\paragraph{Sparse streams.}
We also consider sparse streams, where the underlying frequency vector has at most $k$ non-zero entries (i.e., $\|x\|_0 \leq k$). For this setting, our strategy utilizes a two-stage dimensionality reduction tailored specifically for sparse vectors.

First, we hash the massive $n$-dimensional universe items into $M = \Theta(k/\eps^2)$ random buckets. Because the underlying vector is guaranteed to be $k$-sparse, mapping the non-zero elements into this number of buckets results in very few collisions. This ensures that the $F_2$ moment of the resulting compressed, low-dimensional vector closely preserves the original $F_2$ moment up to a $(1 \pm \eps/4)$ factor. We then compose this with an AMS sketch applied directly to the resulting compressed vector.

To aggressively optimize the space and further remove lingering logarithmic dependencies on the stream length $m$, we do not store the entries of the AMS sketch exactly. The AMS sketch relies on linear combinations with random $\pm 1$ signs, so we conceptually decompose each sketch counter into separate positive and negative updates. We then use Morris Counters \cite{morris1978counting} to separately estimate the positive and negative contributions to each bucket. 

The key to the error analysis relies on the structure of sparse vectors: by the Cauchy-Schwarz inequality, any $k$-sparse stream of length $m$ satisfies $F_2 \geq m^2/k$. By tuning the Morris counters to operate with an error parameter of $\eps' = O(\eps/k)$, the additive error introduced by the approximate counting is tightly bounded, ensuring it is small enough to be absorbed into the overall $(1 \pm \eps)$ multiplicative error. This layered approach roughly allows us to replace the heavy $\log n$ dependence with $\log k$ (alongside a deeply sub-logarithmic $\log\log m$ factor), where $k$ is the sparsity. 
A full statement is given in Theorem~\ref{thm:l2_sparse_improved}.

\section{Preliminaries}
\label{sec:preliminaries}

In this section, we formally define the fundamental concepts, models, and tools used throughout the paper.
We review key concepts from information theory and communication complexity, and briefly describe the standard probability bounds and sketching algorithms utilized in our upper bounds. 

\subsection{Information Theory}

Our space lower bounds rely heavily on information-theoretic measures to quantify how much the internal memory of a streaming algorithm ``knows'' about the input stream. We use standard concepts from Shannon information theory. All logarithms are taken to base $2$. Let $X, Y,$ and $Z$ be discrete random variables.

\begin{definition}[Entropy and Conditional Entropy]
The Shannon entropy of $X$, which measures its expected uncertainty or information content, is defined as:
\[
H(X) := \sum_{x \in \operatorname{supp}(X)} \Pr[X=x] \log \left(\frac{1}{\Pr[X=x]}\right).
\]
The conditional entropy of $X$ given $Y$ is the expected remaining uncertainty of $X$ after $Y$ is observed:
\[
H(X|Y) := \sum_{y \in \operatorname{supp}(Y)} \Pr[Y=y] H(X \mid Y=y) = H(X, Y) - H(Y).
\]
A fundamental property of entropy is that conditioning cannot increase uncertainty: $H(X|Y) \le H(X)$, with equality holding if and only if $X$ and $Y$ are completely independent.
\end{definition}

\begin{definition}[Mutual Information]
The mutual information between $X$ and $Y$ quantifies the amount of information that observing $Y$ reveals about $X$ (and vice versa):
\[
I(X; Y) := H(X) - H(X|Y) = H(Y) - H(Y|X).
\]
By definition, mutual information is symmetric and non-negative: $I(X; Y) \ge 0$. The mutual information between $X$ and $Y$ conditioned on a third variable $Z$ is defined as:
\[
I(X; Y \mid Z) := H(X|Z) - H(X|Y,Z).
\]
\end{definition}

\begin{lemma}[Chain Rules]
Information-theoretic quantities can be decomposed over sequences of random variables using chain rules. For a sequence of random variables $X_1, X_2, \ldots, X_k$, and letting $X_{<i}$ denote the prefix $(X_1, \ldots, X_{i-1})$:
\begin{itemize}
    \item \textbf{Chain Rule for Entropy:} $H(X_1, \ldots, X_k \mid Y) = \sum_{i=1}^k H(X_i \mid X_{<i}, Y)$.
    \item \textbf{Chain Rule for Mutual Information:} $I(X_1, \ldots, X_k ; Y \mid Z) = \sum_{i=1}^k I(X_i ; Y \mid X_{<i}, Z)$.
\end{itemize}
\end{lemma}

\begin{lemma}[Data Processing Inequality]
If $X, Y,$ and $Z$ form a Markov chain $X \to Y \to Z$ (meaning $Z$ is conditionally independent of $X$ given $Y$), then algorithmic post-processing cannot generate new information: $I(X; Z) \le I(X; Y)$. 
\end{lemma}
Because a streaming algorithm's memory state is a function of the stream prefix and internal randomness, the data processing inequality is commonly used to bound the information the algorithm retains about the original input stream.

\subsection{Communication Complexity and Information Cost}

To lower bound the space complexity of a streaming algorithm, we rely on reductions from multi-party communication games. 

\begin{definition}[The Multi-Party One-Way Model]
In the $t$-player one-way communication model, $t$ players $P_1, P_2, \ldots, P_t$ receive respective inputs $X_1, X_2, \ldots, X_t$ drawn from a joint probability distribution $\mu$. A referee may also receive an independent input $y$. The players collaborate to compute a joint function $f(X_1, \ldots, X_t, y)$ or solve a decision problem. 

The communication flows strictly one-way: player $P_1$ computes a message based on their input $X_1$ and sends it to $P_2$. Then, $P_2$ computes a message based on $X_2$ and the message from $P_1$, sending it to $P_3$, and so on. The final player $P_t$ sends a message to the referee, who must output the correct answer. 
\end{definition}

\begin{definition}[Protocols and Transcripts]
A randomized protocol $\Pi$ governs the players' message generation. Players may utilize local private randomness as well as a shared \emph{public randomness} string $R$, which is visible to all participants without incurring communication costs. The \emph{transcript} of the protocol, denoted $\Pi(X)$ (or simply $\Pi$), is the concatenation of all messages sent over the channel. The \emph{communication cost} of $\Pi$ is the worst-case bit-length of its transcript.
\end{definition}

\begin{definition}[Conditional Information Cost]
While traditional communication complexity bounds the raw bit-length of the transcript, \emph{information complexity} measures the actual amount of information the transcript reveals about the inputs. For a protocol $\Pi$ running on inputs $X \sim \mu$ and utilizing public randomness $R$, the \emph{conditional information cost} is defined as the mutual information between the inputs and the transcript, conditioned on the public randomness:
\[
I_\mu(X_1, \ldots, X_t ; \Pi \mid R).
\]
By Shannon's source coding theorem, the expected communication cost of a protocol is always lower-bounded by its information cost: $\mathbb{E}[|\Pi|] \ge I_\mu(X ; \Pi \mid R)$. 
\end{definition}

\textbf{Reduction to streaming space.} 
A streaming algorithm $\mathcal{A}$ evaluated on a stream naturally induces a one-way communication protocol. 
By partitioning the stream into $t$ contiguous blocks assigned to $t$ players, Player $i$ simulates $\mathcal{A}$ on their block and passes the resulting memory state $\mathcal{M}$ as their message to Player $i+1$. 
Thus, proving a lower bound on the information cost of a distributed game strictly lower bounds the space complexity of any streaming algorithm capable of solving that game.

\subsection{Algorithmic Primitives and Probabilistic Tools}

In our upper bound constructions, we utilize the following standard sketching tools and probability structures:

\begin{lemma}[Concentration Inequalities]
We frequently apply standard probability tail bounds to control algorithm failure rates:
\begin{itemize}
    \item \textbf{Markov's Inequality:} For any non-negative random variable $X$ and $a > 0$, $\Pr[X \ge a] \le \frac{\mathbb{E}[X]}{a}$.
    \item \textbf{Chernoff Bound:} For a sum of independent indicator variables $X = \sum_{i=1}^k X_i$, the probability of a large relative deviation is bounded exponentially. For any $\delta > 0$, $\Pr[X \ge (1+\delta)\mathbb{E}[X]] \le \exp\left(-\frac{\delta^2 \mathbb{E}[X]}{2+ \delta}\right)$ and $\Pr[X \le (1-\delta)\mathbb{E}[X]] \le \exp\left(-\frac{\delta^2 \mathbb{E}[X]}{2}\right)$.
\end{itemize}
\end{lemma}

\begin{definition}[$k$-wise Independent Hash Functions]
A family of hash functions $\mathcal{H} = \{h: [n] \to [M]\}$ is \emph{$k$-wise independent} if, for any $k$ distinct elements $x_1, \ldots, x_k \in [n]$ and any target values $y_1, \ldots, y_k \in [M]$, a function $h$ chosen uniformly at random from $\mathcal{H}$ satisfies:
\[
\Pr[h(x_1) = y_1 \land \ldots \land h(x_k) = y_k] = \frac{1}{M^k}.
\]
Storing such a function requires only $O(k \log(\max(n, M)))$ bits, which provides highly memory-efficient pseudo-randomness. We utilize pairwise ($2$-wise) and $4$-wise independent hash functions in our sketches.
\end{definition}

\begin{theorem}[AMS Sketch]
\cite{AlonMS99}
There exists a streaming algorithm that uses $O\left(\frac{1}{\eps^2}\log n\right)$ bits of space and outputs a $(1+\eps)$-approximation to the $F_2$ moment with constant probability. 
\end{theorem}
We remark that the AMS algorithm utilizes a $4$-wise independent hash function $\sigma: [n] \to \{-1, 1\}$ and maintains a linear counter $Z = \sum_{j=1}^m \sigma(a_j) = \sum_{i=1}^n \sigma(i) f_i$, whose square $Z^2$ serves as an unbiased estimator for $F_2$ with bounded variance.

\begin{theorem}[Morris Counters]
\cite{morris1978counting}
There exists a streaming algorithm that outputs a constant-factor approximation to the count of size $n$ with constant probability, using $O(\log\log n)$ bits of space. 
\end{theorem}
To track a count up to $N$, a Morris counter explicitly stores a logarithmic value $c$ (requiring only $O(\log \log N)$ bits). 
To update the counter, $c$ is incremented with probability $2^{-c}$, giving an unbiased estimator for the true count.

\begin{theorem}[$F_0$ Estimators]
\cite{kane2010optimal,Blasiok20}
\label{thm:fzero:monitor}
There exists a streaming algorithm that uses $O\left(\frac{1}{\eps^2}\log\frac{1}{\delta}+\log n\right)$ bits of space and with constant probability, outputs a $(1+\eps)$-approximation to the $F_0$ moment at all points in the stream. 
\end{theorem}

\section{Insertion-Only Lower Bounds for \texorpdfstring{$F_2$}{F2} Estimation}
\label{sec:insertion}

We present a strengthened space lower bound for the problem of $F_2$ estimation in the insertion-only model, incorporating the optimal dependence on the failure probability $\delta$. The proof strategy closely follows the multi-scale framework introduced by \cite{BravermanZ25}, but replaces their underlying communication problem with a robust variant that accounts for error probability.

\thmmain*

We utilize standard concepts from information theory. For a communication protocol $\Pi$ with inputs $X$ drawn from a distribution $\mu$ and public randomness $R$, the conditional information cost is the mutual information $I(X; \Pi(X) | R)$ between the input and the protocol transcript, conditioned on the public randomness. The communication cost of any protocol is an upper bound on its information cost. For a streaming algorithm processing a random stream $X$, we analyze the information $I(X; \mathcal{M})$ that the memory state $\mathcal{M}$ contains about the stream.

\subsection{The Exam Mostly Set Disjointness Problem}

The core of the streaming space lower bound in \cite{BravermanZ25} relies on a communication game called ``Exam Set Disjointness'' (EDISJ). However, EDISJ is highly sensitive to noise; a protocol can fail to identify a unique intersecting element by simply making an error on a single critical coordinate. Consequently, it cannot yield the robust $\Omega(\log(1/\delta))$ information cost bounds required to establish optimal dependencies on the failure probability $\delta$. 
To overcome this limitation, we introduce a noise-robust variant inspired by the ``Mostly Set Disjointness'' problem introduced by Kamath et al.\ \cite{KPW21}. We first define a technical relaxation of disjointness that permits a bounded amount of noise in the sets.

\begin{definition}[$\kappa$-Almost-Disjointness]
    A collection of sets $S_1, \ldots, S_t \subseteq U$ is said to be $\kappa$-almost-disjoint if the total number of non-unique element memberships is at most $\kappa$. Formally, the number of pairs $(i, x)$ where $x$ appears in $S_i$ and at least one other set $S_j$ is bounded by $\kappa$:
    \[
    \big|\{ (i, x) \in [t] \times U : x \in S_i \text{ and } x \in S_j \text{ for some } j \neq i \}\big| \leq \kappa.
    \]
\end{definition}

With this relaxation, we define our core communication primitive.

\begin{definition}[Exam Mostly Set Disjointness (EMostlyDISJ)]
    The setup involves $t$ players and one referee. Let $U$ be a universe of items. The inputs are sets $S_1, \ldots, S_t \subseteq U$ for the players and an element $j \in U$ for the referee. The input sets are promised to satisfy exactly one of two conditions:
    \begin{itemize}
        \item \textbf{NO Instance:} The sets $S_1, \ldots, S_t$ are globally $\kappa$-almost-disjoint.
        \item \textbf{YES Instance:} There is a unique element $j_0 \in U$ that is common to at least $ct$ of the sets for some constant $c \in (0,1)$, and the remaining elements across the sets are $\kappa$-almost-disjoint.
    \end{itemize}
    The communication is one-way from player $1$ to player $t$, and finally to the referee. With failure probability at most $\delta$, the referee must decide if the input is a YES instance and the heavily intersecting element $j_0$ is exactly equal to the referee's target element $j$.
\end{definition}

\subsubsection{The One-Bit Problem: \texorpdfstring{$F_{ct,t}$}{Fctt}}

To establish the information complexity of EMostlyDISJ, we reduce it to a primitive one-bit problem, much like how standard Set Disjointness reduces to the multi-party AND function. 

In the $F_{ct,t}$ problem, each of the $t$ players receives a single bit $Y_i \in \{0,1\}$. Their goal is to distinguish between two cases regarding the Hamming weight of the joint input vector $Y=(Y_1, \ldots, Y_t)$: a NO instance where the Hamming weight is at most $1$, and a YES instance where the Hamming weight is at least $ct$.

We analyze the information cost of this problem with respect to a hard distribution $\nu$ over NO instances, adapting the approach of \cite{KPW21}. By bounding the information revealed on NO instances, we can embed many instances of this problem into a continuous stream.

\begin{definition}[Hard Distribution $\nu$]
    Let $R$ be a public random variable drawn uniformly from the player indices $[t] = \{1, 2, \ldots, t\}$. The distribution $\nu$ over the joint inputs $Y \in \{0,1\}^t$ is defined conditionally based on $R$: if $R=i$, the input $Y$ is the all-zeros vector $0^t$ with probability $1/2$, and the standard basis vector $e_i$ (where only the $i$-th coordinate is $1$) with probability $1/2$.
\end{definition}

\begin{lemma}[Conditional Information Cost on $\nu$, adapted from \cite{KPW21}]\label{lem:sparse_cost_formal}
Let $\Pi$ be a one-way communication protocol that solves $F_{ct,t}$ with failure probability at most $\delta$. 
Let the inputs $Y$ be drawn from the distribution $\nu$. 
The conditional information cost of the protocol with respect to $\nu$ is lower-bounded by:
\[I_\nu(Y; \Pi \mid R) \geq \Omega\left(\min\left(1,\frac{1}{t} \log\left(\frac{1}{\delta}\right)\right)\right).\]
\end{lemma}
\begin{proof}
Assume without loss of generality that $\delta \le 1/4$. 
Lemma 3.7 of \cite{KPW21} establishes a constant lower bound for protocols with sufficiently small error. 
Specifically, if a protocol has a failure probability $\delta'$ that satisfies $t \leq c' \log(1/(2e\delta'))$ for some absolute constant $c'$, then its conditional information cost is bounded below by a constant $I_\nu(Y;\Pi \mid R, C) \geq \Omega(1)$,  where $C$ represents the public randomness. 

If the base protocol's failure probability $\delta$ already satisfies this condition ($t \leq c' \log(1/(2e\delta))$), then no amplification is needed. The lower bound is $\Omega(1)$ directly.

To handle arbitrary failure probabilities $\delta$, we use an information-theoretic amplification argument. 
Suppose we have a base protocol $\Pi$ that fails with probability $\delta$. We construct a boosted protocol $\Pi'$ by executing $\Pi$ independently $r$ times using independent private randomness and independent public coins $C_1, \dots, C_r$, and having the referee output the majority vote of the results. 
By a standard Chernoff bound, the failure probability of the boosted protocol $\Pi'$ is bounded by $\delta' \leq (4\delta)^{r/2}$.

We choose the number of repetitions to be $r = \Theta\left(\frac{t}{\log(1/\delta)}\right)$. 
This ensures that the amplified error $\delta'$ is driven down precisely to the threshold required by the condition from \cite{KPW21}. 
The transcript of the boosted protocol $\Pi'$ simply consists of $r$ independent transcripts from the base protocol, denoted $(\Pi^1, \ldots, \Pi^r)$. 
By applying the established lower bound for the small-error regime, we know the information cost of $\Pi'$ is bounded below by a constant:
\[\Omega(1) \leq I_\nu(\Pi^1, \ldots, \Pi^r; Y \mid R).\]
We now decompose this joint information using the chain rule for mutual information:
\[\Omega(1) \le \sum_{i=1}^r I_\nu(\Pi^i; Y \mid R, C_1, \dots, C_r, \Pi^{<i}),\]
where $\Pi^{<i}$ denotes the sequence of transcripts from the first $i-1$ repetitions. 
By an averaging argument, there must exist at least one index $i \in \{1, \ldots, r\}$ that contributes at least the average amount of information to the sum:
\[I_\nu(\Pi^i; Y \mid R, C_1, \dots, C_r, \Pi^{<i}) \ge \frac{\Omega(1)}{r} = \Omega\left(\frac{1}{t}\log\frac{1}{\delta}\right).\]
Finally, we decouple this term from the transcripts of the previous repetitions.
By expanding the conditional mutual information into Shannon entropies, we have:
\begin{align*}
I_\nu(\Pi^i; Y \mid R, C_1, \dots, C_r, \Pi^{<i}) &= H(\Pi^i \mid R, C_1, \dots, C_r, \Pi^{<i}) - H(\Pi^i \mid R, C_1, \dots, C_r, \Pi^{<i}, Y) \\
&\le H(\Pi^i \mid R, C_i, \Pi^{<i}) - H(\Pi^i \mid R, C_i, Y).
\end{align*}
The last inequality holds because the transcripts of separate repetitions are generated using independent private randomness and public coins; thus, conditioned on the exact inputs $Y$, public randomness $C_i$ and the selector $R$, $\Pi^i$ is entirely independent of the previous transcripts $\Pi^{<i}$ and other public coins. 
Next, we use the fundamental property that conditioning reduces entropy, which implies $H(\Pi^i \mid R, C_i, \Pi^{<i}) \le H(\Pi^i \mid R, C_i)$. 
Substituting this upper bound into our equation yields:
\[I_\nu(\Pi^i; Y \mid R, C_1, \dots, C_r, \Pi^{<i}) \le H(\Pi^i \mid R, C_i) - H(\Pi^i \mid R, C_i, Y) = I_\nu(\Pi^i; Y \mid R, C_i).\]
Because $\Pi^i$ is identical in distribution to a single invocation of our original protocol $\Pi$ utilizing public coins $C$, we conclude that the conditional information cost of the base protocol itself satisfies:
\[I_\nu(\Pi; Y \mid R, C) \ge \Omega\left(\frac{1}{t}\log\frac{1}{\delta}\right).\]
This completes the proof.
\end{proof}

\subsubsection{From Hard Distribution to Product Distribution}
To embed instances of the EMostlyDISJ communication game into a continuous streaming algorithm lower bound, we must transition from the artificially correlated hard distribution $\nu$ to a fully independent product distribution. 
We define $\mu_p$ as the product distribution over $\{0,1\}^t$ where each bit $Y_i$ is set to $1$ independently with probability $p$. 
Because data stream elements appear independently, simulating the stream naturally generates inputs drawn from $\mu_p$.

The following lemma, adapting a crucial technique from \cite{BravermanZ25}, establishes that the information cost lower bound holds even under this independent product distribution, scaling proportionately with $p$.

\begin{lemma}[Adapted from Lemma 4.4 of \cite{BravermanZ25}]\label{lem:info_relation_formal}
Let $\Pi$ be a communication protocol for $F_{ct,t}$ with failure probability at most $\delta$. 
Let $p < 1/t$. For the independent product distribution $\mu_p$, the information cost of the protocol is lower-bounded by:
\[I_{\mu_p}(Y; \Pi) \geq \Omega\left(p\cdot \min\left(t,\log\left(\frac{1}{\delta}\right)\right)\right).\]
\end{lemma}
\begin{proof}
We shall define a probabilistic coupling between the product distribution $\mu_p$ and the hard distribution $\nu$. 
We achieve this by introducing an auxiliary indicator variable $D \in \{0,1\}$ that acts as a filter, where conditioning on $D=1$ forces $Y \sim \mu_p$ to perfectly emulate $Y \sim \nu$.

\paragraph{Step 1: Bounding probabilities under $\mu_p$.}
First, we evaluate the probabilities of the critical NO instances (the all-zeros vector $0^t$ and the standard basis vectors $e_i$) under the product distribution $\mu_p$. Since $p < 1/t$, the probability of observing the all-zeros vector is bounded below by a constant:
\[\mu_p(0^t) = (1-p)^t \geq \left(1 - \frac{1}{t}\right)^t \geq \frac{1}{4} \quad (\text{for } t \geq 2).\]
Similarly, the probability of observing any specific standard basis vector $e_i$ is:
\[\mu_p(e_i) = p(1-p)^{t-1} \geq p \left(1 - \frac{1}{t}\right)^{t-1} \geq \frac{p}{e} \geq \frac{p}{4}.\]

\paragraph{Step 2: Constructing the coupling variable $D$.}
Recall that under the marginalized hard distribution $\nu$, the probabilities are $\nu(0^t) = 1/2$ and $\nu(e_i) = 1/(2t)$. 
Because the probabilities of these events under $\mu_p$ strictly dominate their relative frequencies under $\nu$ (up to a scaling factor), we can define a boolean random variable $D$ as a randomized function of $Y \sim \mu_p$ such that:
\begin{enumerate}
\item 
Conditioned on $D=1$, the distribution of $Y$ is exactly the marginal distribution of $\nu$.
\item 
The probability of the coupling succeeding is maximized: $\Pr[D=1] = \Theta(p \cdot t)$.
\end{enumerate}
By Bayes' theorem, ensuring $(Y \mid D=1) \sim \nu$ requires setting the conditional probabilities as:
\[\Pr[D=1 \mid Y=y] = \frac{\nu(y) \Pr[D=1]}{\mu_p(y)}.\]
This is a valid probability assignment as long as the ratio on the right-hand side never exceeds $1$. 
For $y = 0^t$, the ratio $\nu(0^t)/\mu_p(0^t)$ is upper bounded by $(1/2) / (1/4) = 2$. 
For $y = e_i$, the ratio $\nu(e_i)/\mu_p(e_i)$ is upper bounded by $(1/(2t)) / (p/4) = 2/(pt)$. 
Thus by setting $\Pr[D=1] = pt/2$, it follows that all conditional probabilities are contained in $[0,1]$ (since $p < 1/t$, the probability for $0^t$ evaluates to $pt < 1$, and for $e_i$ evaluates to a constant), confirming that $\Pr[D=1] = \Theta(p \cdot t)$. 
For all other vectors $y \notin \{0^t, e_1, \ldots, e_t\}$, we simply set $\Pr[D=1 \mid Y=y] = 0$.

Because $D$ is generated solely based on the realized value of $Y$, the variables form a Markov chain $D \to Y \to \Pi$. 
This implies that, conditioned on $Y$, the protocol transcript $\Pi$ is entirely independent of $D$, yielding $I_{\mu_p}(D; \Pi \mid Y, C) = 0$.

\paragraph{Step 3: Chain rule expansion for $D$.}
We now use the chain rule for mutual information to relate the information cost under $\mu_p$ to the cost conditioned on the coupling success. 
We expand the joint mutual information $I_{\mu_p}(Y, D; \Pi)$ in two different ways:
\begin{align*}
I_{\mu_p}(Y; \Pi \mid C) &= I_{\mu_p}(Y, D; \Pi \mid C) - I_{\mu_p}(D; \Pi \mid Y, C) \\
&= I_{\mu_p}(Y, D; \Pi \mid C) \quad \text{(by the Markov property)} \\
&= I_{\mu_p}(D; \Pi \mid C) + I_{\mu_p}(Y; \Pi \mid D, C)
\end{align*}
By the non-negativity of mutual information, $I_{\mu_p}(D; \Pi \mid C) \ge 0$. 
Thus, we can expand the conditional mutual information over the states of $D$:
\begin{align*}
I_{\mu_p}(Y; \Pi \mid C) &\ge I_{\mu_p}(Y; \Pi \mid D, C) \\
&= \Pr[D=0] I_{\mu_p}(Y; \Pi \mid D=0, C) + \Pr[D=1] I_{\mu_p}(Y; \Pi \mid D=1, C) \\
&\ge \Pr[D=1] I_{\mu_p}(Y; \Pi \mid D=1, C).
\end{align*}
Because the distribution of $Y$ conditioned on $D=1$ is exactly $\nu$, the information cost in this branch is precisely $I_\nu(Y; \Pi \mid C)$. 
Substituting $\Pr[D=1] = \Theta(p \cdot t)$, we have:
\[I_{\mu_p}(Y; \Pi \mid C) \ge \Theta(p \cdot t) \cdot I_\nu(Y; \Pi \mid C).\]

\paragraph{Step 4: Chain rule expansion for $R$.}
We now have a bound in terms of the unconditional information cost $I_\nu(Y; \Pi \mid C)$. 
However, our base lower bound from Lemma \ref{lem:sparse_cost_formal} bounds the \emph{conditional} information cost $I_\nu(Y; \Pi \mid R, C)$. 
    
To bridge this gap, we observe the relationship between $Y$ and the public randomness $R$ under $\nu$. 
If $Y = e_i$, the chosen index must be $R = i$. If $Y = 0^t$, $R$ is uniformly distributed in $[t]$. 
Therefore, $R$ can be viewed as a randomized function of $Y$. 
Consequently, $R \to Y \to \Pi$ forms another Markov chain, meaning $R$ provides no additional information about the transcript $\Pi$ once $Y$ and $C$ are known. 
Thus, $I_\nu(R; \Pi \mid Y, C) = 0$.

Applying the chain rule again:
\begin{align*}
I_\nu(Y; \Pi \mid C) &= I_\nu(Y, R; \Pi \mid C) - I_\nu(R; \Pi \mid Y, C) \\
&= I_\nu(Y, R; \Pi \mid C) \\
&= I_\nu(R; \Pi \mid C) + I_\nu(Y; \Pi \mid R, C) \\
&\ge I_\nu(Y; \Pi \mid R, C).
\end{align*}

\paragraph{Step 5: Putting it all together.}
We concatenate the inequalities derived in Step 3 and Step 4, and substitute the conditional lower bound from Lemma \ref{lem:sparse_cost_formal}:
\begin{align*}
I_{\mu_p}(Y; \Pi \mid C) &\ge \Theta(p \cdot t) \cdot I_\nu(Y; \Pi \mid R, C) \\
&\ge \Theta(p \cdot t) \cdot \Omega\left(\min\left(1,\frac{1}{t}\log\frac{1}{\delta}\right)\right) \\
&= \Omega\left(p\cdot\min\left(t,\log\frac{1}{\delta}\right)\right).
\end{align*}
This establishes the required information cost bound for the independent product distribution.
\end{proof}

\subsubsection{Direct Sum for EMostlyDISJ}
We now apply an information-theoretic direct sum argument, following the framework of Braverman and Zamir \cite{BravermanZ25}, to lift the $\Omega(p \log(1/\delta))$ lower bound from the independent 1-bit $F_{ct,t}$ instances to the full EMostlyDISJ problem. 
Let $m$ be the total set size parameter across all players. 
To suppress spurious collisions and simulate disjointness, we assume the abstract universe $U$ is exceptionally large, specifically $|U| \geq m^4$. 
We set the independent inclusion probability for our product distribution to $p = m/(2t|U|)$. 
Before executing the direct sum, we must guarantee that sampling sets independently under $\mu_p$ naturally generates inputs that satisfy the NO instance promise of EMostlyDISJ: namely, that the resulting sets are $\kappa$-almost-disjoint with overwhelmingly high probability.

\begin{proposition}
\label{prop:almost_disjoint}
Let $U$ be a universe of size $|U| \ge m^4$. 
Suppose $t$ players construct sets $S_1, \ldots, S_t$ such that each $S_i$ is chosen uniformly and independently at random from all subsets of $U$ of size exactly $m/t$. 
With probability at least $1-\delta$, the resulting collection of sets is $\kappa$-almost-disjoint, where $\kappa = O(\log(1/\delta))$. 
\end{proposition}
\begin{proof}
We upper bound the probability of having too many non-disjoint coordinates, which we call an \emph{overlap} and occurs when an element $u \in U$ is chosen by two or more players. 
Because each of the $t$ players draws a set of size exactly $m/t$, the total number of slots draws across all players is exactly $m$. 
For $\kappa$ slots out of the $m$ total slots to be non-unique, they must form at least $\kappa/2$ colliding pairs. 
We construct a union bound over all possible ways this could occur:
\begin{enumerate}
\item 
There are $\binom{m}{\kappa}$ ways to select which $\kappa$ specific slots will be involved in the collisions.
\item 
For each of these $\kappa$ slots, we must specify a collision partner from the $m$ total slots, giving at most $m^\kappa$ possible partner assignments.
\item 
Forming these partner assignments forces at least $\kappa/2$ equality constraints between the slots. 
We upper bound the probability of satisfying a specific set of $\kappa/2$ equalities. 
The total number of valid ways to assign elements to the $m$ slots, drawing $s$ items without replacement for each of the $t$ players, is exactly $\left(\frac{|U|!}{(|U|-s)!}\right)^t \ge (|U|-s)^m$. 
The number of assignments satisfying the $\kappa/2$ equalities is at most $|U|^{m-\kappa/2}$, because the equalities group the $m$ slots into at most $m-\kappa/2$ connected components, and each component can take at most $|U|$ values. 
Thus, the probability of satisfying a specific set of $\kappa/2$ equalities is at most:
\[\frac{|U|^{m-\kappa/2}}{(|U|-s)^m} = \frac{1}{|U|^{\kappa/2}} \left( \frac{|U|}{|U|-s} \right)^m.\]
Since $|U| \ge m^4$ and $s \le m$, we have $\left( \frac{|U|}{|U|-s} \right)^m = \left( 1 + \frac{s}{|U|-s} \right)^m \le e^{sm / (|U|-s)}$. 
Since $|U|-s \ge m^4-m \ge m^2$, the exponent is at most $sm/(|U|-s) \le m^2/m^2 = 1$. 
Therefore, the probability of satisfying any specific set of $\kappa/2$ equality constraints is at most $e / |U|^{\kappa/2}$.
\end{enumerate}
Multiplying these factors together, the probability of observing at least $\kappa$ overlaps is upper bounded by
\[\binom{m}{\kappa} m^\kappa \frac{e}{|U|^{\kappa/2}} \le e \left(\frac{em}{\kappa}\right)^\kappa \left(\frac{m^2}{|U|}\right)^{\kappa/2} = e \left(\frac{e m^2}{\kappa \sqrt{|U|}}\right)^\kappa.\]
Substituting our universe size assumption $|U| \geq m^4$ implies $\sqrt{|U|} \geq m^2$. 
Thus, the probability is at most $e \left(\frac{e}{\kappa}\right)^\kappa$. 
By setting $\kappa = C \log(1/\delta)$ for a sufficiently large constant $C$, it follows that the probability is at most $O(\delta)$. 
Thus, the uniformly sampled fixed-size sets are globally $\kappa$-almost-disjoint with probability at least $1-\delta$.
\end{proof}

With the promise condition satisfied, we establish the final communication lower bound by synthesizing $|U|$ parallel games.

\begin{corollary}
\label{cor:comm_cost_final}
Let $\kappa = O(\log(1/\delta))$ and assume $t=\Omega\left(\log\frac{1}{\delta}\right)$ and $t e^{-m/(6t)} \le \delta$. 
Let $\Pi$ be a one-way communication protocol of worst-case transcript length $C_\Pi$, utilizing public coins $C$, that solves the $t$-party EMostlyDISJ problem with failure probability at most $O(\delta)$, total set size parameter $m$, almost-disjointness parameter $\kappa$, and universe size $|U| \geq m^4$. 
Let $\beta := \Pr[F=1] \leq t e^{-m/(6t)}$ be the abort probability. 
The conditional information cost of the protocol satisfies:
\[I_\mu(\Pi; X' \mid C) \geq \frac{\Omega\left(\frac{m}{t}\log\left(\frac{1}{\delta}\right)\right) - h_2(\beta) - \beta(C_\Pi + 1)}{1-\beta},\]
where $h_2$ is the binary entropy function.
\end{corollary}
\begin{proof}
We employ the direct sum reduction formalized in Section 4.2 of \cite{BravermanZ25}. 
We consider $|U|$ parallel, independent instances of the 1-bit problem $F_{ct,t}$. 
Let the input matrix $Y \in \{0,1\}^{t \times |U|}$ consist of columns $Y^j$ drawn i.i.d. from the product distribution $\mu_p$, where $p = m/(2t|U|)$. 
    
We construct a protocol $\Pi'$ that solves these parallel instances by simulating the EMostlyDISJ protocol $\Pi$. 
Each player $i$ locally constructs a preliminary set $S'_i = \{j \in U : Y^j_i = 1\}$. 
Because $Y_i^j \sim \text{Bern}(p)$, the expected size of this set is $\mathbb{E}[|S'_i|] = p|U| = m/(2t)$. 
    
If any player's set size exceeds double its expectation ($|S'_i| > m/t$), the reduction immediately aborts and declares a failure (Event $F$). 
Otherwise, player $i$ pads $S'_i$ up to a size of exactly $m/t$ using dummy elements drawn uniformly at random without replacement from the unselected elements of the common universe $U \setminus S'_i$. 
Let $X' = (S_1, \ldots, S_t)$ be the resulting collection of padded sets. 
The players then execute $\Pi(X')$.

\paragraph{Step 1: Bounding failure and error.}
By a standard Chernoff bound, the probability that a single player's set exceeds $m/t$ is at most $e^{-m/(6t)}$. 
By a union bound over all $t$ players, the probability of the failure event $F$ is bounded by $\beta := \Pr[F=1] \leq t e^{-m/(6t)}$. 
By adjusting constants appropriately in our parameter assumptions, $\beta \leq \delta$. 
Provided $F$ does not occur, Lemma 4.6 of \cite{BravermanZ25} guarantees that the resulting compound protocol $\Pi'$ correctly solves the embedded instances with high probability. 
The total error probability of $\Pi'$ is upper bounded by $O(\delta)$. 

Crucially, because the initial Bernoulli process treats all elements of $U$ symmetrically, padding uniformly from the unselected elements guarantees that the resulting set $S_i$ is distributed uniformly over all subsets of $U$ of size $m/t$. 
Because the players pad independently, conditioned on $F=0$, the collection $X'$ matches the distribution $\mu$ of independent, uniform fixed-size subsets. 
Because the universe size is $|U| \ge m^4$, Proposition \ref{prop:almost_disjoint} applies directly to the padded sets $X'$. 
Thus, with probability at least $1 - \delta$, the constructed sets satisfy the $\kappa$-almost-disjointness promise.

\paragraph{Step 2: Direct sum over independent instances.}
Because the columns $Y^j$ are drawn completely independently under $\mu_p$, the mutual information between the full matrix $Y$ and the transcript $\Pi'$ decomposes linearly via the super-additivity of independent inputs:
\[I(Y; \Pi') \geq \sum_{j \in U} I(Y^j; \Pi').\]
Applying our 1-bit lower bound from Lemma \ref{lem:info_relation_formal} to each instance, we aggregate the information:
\[I(Y; \Pi') \geq |U| \cdot \Omega\left(p \log\left(\frac{1}{\delta}\right)\right) = |U| \cdot \Omega\left(\frac{m}{2t|U|} \log\left(\frac{1}{\delta}\right)\right) = \Omega\left(\frac{m}{t} \log\left(\frac{1}{\delta}\right)\right).\]

\paragraph{Step 3: Relating to the original protocol.}
Finally, we must relate the unconditional information cost $I(Y; \Pi')$ of the simulated inputs back to the actual conditional information cost $I_\mu(\Pi; X') = I(\Pi'; X' \mid F=0)$ of the base protocol operating on the padded sets. 

Conditioned on $F=0$, the symmetric random padding from the common universe $U$ ensures the padded sets match the distribution $\mu$ of independent uniform fixed-size subsets. 
Moreover, conditioned on $X'$, $F=0$, and the public coins $C$, the transcript is generated only from $X'$ and the protocol's private coins. 
Hence, $Y \to X' \to \Pi'$ is a Markov chain, which implies by the data processing inequality:
\[I(Y; \Pi' \mid F=0, C) \leq I(X'; \Pi' \mid F=0, C) = I_\mu(\Pi; X' \mid C).\]

We expand the mutual information $I(Y; \Pi' \mid C)$ using the chain rule over the abort indicator $F$:
\begin{align*}
I(Y; \Pi' \mid C) &\leq I(Y; \Pi', F \mid C) \\
&= I(Y; F \mid C) + I(Y; \Pi' \mid F, C) \\
&\leq h_2(\beta) + (1-\beta)I(Y; \Pi' \mid F=0, C) + \beta I(Y; \Pi' \mid F=1, C).
\end{align*}
Let $C_\Pi$ be the worst-case bit length of the transcript of $\Pi$. 
Even on an abort ($F=1$), we can pad the compound transcript to length $C_\Pi + 1$. 
Thus, the conditional mutual information $I(Y; \Pi' \mid F=1)$ is upper bounded by $C_\Pi + 1$. 
Substituting these bounds into our expansion, we have
\[I(Y; \Pi' \mid C) \leq h_2(\beta) + (1-\beta)I_\mu(\Pi; X' \mid C) + \beta(C_\Pi + 1).\]
Rearranging the terms to isolate $I_\mu(\Pi; X')$, we have
\[I_\mu(\Pi; X' \mid C) \geq \frac{I(Y; \Pi' \mid C) - h_2(\beta) - \beta(C_\Pi + 1)}{1-\beta}.\]

Substituting the main term $I(Y; \Pi' \mid C) \geq \Omega\left(\frac{m}{t}\log\left(\frac{1}{\delta}\right)\right)$ from Step 2, we obtain the lower bound
\[I_\mu(\Pi; X' \mid C) \geq \frac{\Omega\left(\frac{m}{t}\log\left(\frac{1}{\delta}\right)\right) - h_2(\beta) - \beta(C_\Pi + 1)}{1-\beta}.\]
This concludes the direct sum. 
\end{proof}

\subsection{\texorpdfstring{$F_2$}{F2} Estimation Lower Bound via Multi-Scale Reduction}

We now employ the multi-scale reduction framework of \cite{BravermanZ25} (Section 5) using EMostlyDISJ as the building block.

\subsubsection{Reduction from EMostlyDISJ to \texorpdfstring{$F_2$}{F2}}

We first generalize the reduction in Section 5.1 of \cite{BravermanZ25} to EMostlyDISJ and allow for flexible parameters $t$. Note that this section describes the construction for a \emph{single level} (i.e., a single instance of the game played with exactly $t$ players); in the next section, we will embed multiple such levels simultaneously. Fix $n, \eps$. Let $2 \leq t \leq \eps\sqrt{n}/2$. 
Define the super-element size $d := \lfloor \frac{\eps^2 n}{4t^2} \rfloor \geq 1$. 

We reduce from a single instance of EMostlyDISJ$_t$. The abstract universe of items for this game is instantiated as $U^d$ (where each abstract item is a super-element, defined as a $d$-tuple of distinct elements from a base stream universe $U$). The total set size parameter $m$ from the previous section (which bounds the total number of abstract items distributed across all players' sets) is set to $m = \lfloor n/(4d) \rfloor = \Theta(t^2/\eps^2)$. We set the almost-disjointness parameter to $\kappa = O(\log(1/\delta))$.

Given an instance $X=(S_1, \ldots, S_t, x) \in (U^d)^t \times U^d$, we construct a stream $s(X)$ of base items from $U$. 
This implicitly defines a frequency vector $v \in \mathbb{N}^{|U|}$, where each coordinate $v_u$ records the exact number of times base item $u \in U$ appears in the stream. 
First, we concatenate the elements of the players' sets $S_i$: each super-element is ``unzipped'' and its $d$ distinct constituent base elements from $U$ are written to the stream. 
Since there are at most $m$ super-elements across all players, this contributes $md \leq n/4$ base items to the stream. 
Finally, we append the referee's target super-element $x$ by writing its $d$ base elements repeatedly $k = \lceil t/\eps \rceil$ times. 
Because $k = \lceil t/\eps \rceil$ and $d \le \frac{\eps^2 n}{4t^2}$, we have $dk \approx \frac{\eps n}{4t} \le n/4$. 
The total stream length is upper bounded by $md + dk \le n/4 + n/4 \le n$.

\begin{definition}[$\zeta$-Base-Separation]
A stream built from an instance $X$ is $\zeta$-base-separated if the number of non-unique base-element memberships arising from coordinates of distinct abstract items is at most $\zeta$.
\end{definition}

\begin{lemma}[Adapted from Lemma 5.2 of \cite{BravermanZ25}]
\label{lem:f2_reduction_mostly}
If $t \geq C_1 \log\frac{1}{\delta}$ and $\kappa\le C_2\log\frac{1}{\delta}$ for sufficiently large constants $C_1>C_2$, and the stream is $\zeta$-base separated for $\zeta\le C_2\log\frac{1}{\delta}$, then a $(1\pm\Theta(\eps))$-approximation to $F_2(s(X))$ solves any valid instance $X$ of EMostlyDISJ$_t$ where the total number of base element collisions is at most $\kappa$.
\end{lemma}
\begin{proof}
We analyze the gap in the second frequency moment $F_2(v) = \sum_{u \in U} v_u^2$ between YES and NO instances. 
For simplicity, we assume $c=1/2$ and $\gamma\le\kappa$. 
Let $F_{bg}$ denote the $F_2$ contribution from the background base items (all items not belonging to the referee's target $x$ or the heavy element $j_0$). 
By assumption, the total number of base element collisions is at most $\kappa$. 
Because the $m$ super-elements contribute $md \leq n/4$ base items, then we have 
\[F_{bg} \le d(\kappa+1)^2 + n/4 \le n/4 + \mathcal{O}(d\kappa^2).\]
We now distinguish the YES instance from two logically distinct NO instances.

In the first NO case, the input sets are $\kappa$-almost disjoint, and the referee's target $x$ is repeated $k = \lceil t/\eps \rceil$ times. 
The target $x$ appears at most $\kappa$ times among the players. 
Including additional collisions, the $d$ coordinates of the target reach a frequency of at most $k + 2\kappa$. 
Hence, we have 
\[F_2(\text{NO}_1) \le F_{bg} + d(k + 2\kappa)^2 \le n/4 + d k^2 + O(dk\kappa + d\kappa^2).\]

In the second NO case, there is a heavy element $j_0$ appearing $t/2$ times, but the referee's target is a different element $x$ appearing $k$ times. 
Because $x$ and $j_0$ are distinct super-elements, their base coordinates are completely disjoint except for the additional collisions, which occur at most $\kappa$ times. 
Thus, the $d$ coordinates of $j_0$ have frequency at most $t/2 + \kappa$, and the $d$ coordinates of $x$ have frequency at most $k + \kappa$. 
Hence, we have
\[F_2(\text{NO}_2) \le F_{bg} + d(t/2 + \kappa)^2 + d(k + \kappa)^2 \le n/4 + d\left( \frac{t^2}{4} + k^2 \right) + O(dt\kappa + dk\kappa + d\kappa^2).\]

In a YES instance, the heavy element $j_0$ is exactly the referee's target $x$. 
The $d$ target coordinates appear at least $t/2$ times from the players and $k$ times from the referee, reaching a frequency of $t/2 + k$.
Then we have
\[F_2(\text{YES}) \geq F_{bg}' + d \cdot (t/2+k)^2 \geq n/4 - d(t/2) + d \cdot (t/2+k)^2.\]
Here, $F_{bg}' \ge n/4 - d(t/2)$ accounts for the target items being removed from the background calculation, assuming the total size of the players' sets is exactly $m$. 

Observe that the $F_2$ moment for the NO instance is maximized in Case 2, and so the additive gap between the YES case and the NO case is at least:
\begin{align*}
\Delta F_2 &= F_2(\text{YES}) - F_2(\text{NO}_2) \\
&\geq d \cdot [(t/2+k)^2 - (t/2+\kappa)^2 - (k+\kappa)^2] - O(dt\kappa + dk\kappa + d\kappa^2) - dt/2 \\
&= d \cdot [ (t/2)^2 + tk + k^2 - (t/2)^2 - t\kappa - \kappa^2 - k^2 - 2k\kappa - \kappa^2 ] - O(dt\kappa + dk\kappa + d\kappa^2) - dt/2 \\
&= d \cdot [ tk - O(t\kappa) - O(k\kappa) - O(\kappa^2) - O(t)].
\end{align*}
Substituting $k = \lceil t/\eps \rceil$ and $d = \lfloor \eps^2 n / (4t^2) \rfloor$, we have
\[\Delta F_2 \geq \frac{\eps^2 n}{4t^2} \cdot \left[ \frac{t^2}{\eps} - O\left(\frac{\kappa t}{\eps}\right) \right].\]
The dominant term is $\frac{\eps^2 n}{4t^2} \cdot \frac{t^2}{\eps} = \frac{\eps n}{4}$.
The subtracted negative terms are negligible, since $t^2/\eps \gg \kappa t / \eps$ for $t \geq C_1 \log\frac{1}{\delta}$ and $\kappa\le C_2\log\frac{1}{\delta}$ for sufficiently large constants $C_1>C_2$.  
We also need $t^2/\eps \gg t^2$, which holds for $\eps < 1$. If $t \geq C \log(1/\delta)$, the additive gap remains $\Delta F_2 = \Omega(\eps n)$.

To establish that this additive gap translates to a strict $(1+\Omega(\eps))$-multiplicative factor, we bound the maximum value of $F_2(\text{NO}_2)$. 
By our parameter choices, the referee's isolated squared contribution is $dk^2 \approx \left(\frac{\eps^2 n}{4t^2}\right)\frac{t^2}{\eps^2} = \Theta(n)$. 
Since $F_{bg} \le n/4 + O(\kappa^2)$, it follows that 
\[F_2(\text{NO}_2) \le n/4 + \Theta(n) + o(n) = \Theta(n).\]
Since the baseline is $F_2(\text{NO}_2) = \Theta(n)$ and the additive gap is $\Delta F_2 = \Omega(\eps n)$, the multiplicative ratio between the instances is:
\[\frac{F_2(\text{YES})}{F_2(\text{NO}_2)} \geq \frac{F_2(\text{NO}_2) + \Omega(\eps n)}{F_2(\text{NO}_2)} = 1 + \Omega\left(\frac{\eps n}{n}\right) = 1 + \Omega(\eps).\]
Because the YES and NO cases differ by a $(1+\Omega(\eps))$-multiplicative factor, any streaming algorithm providing a $(1 \pm c\eps)$-relative approximation to $F_2$ (for a sufficiently small constant $c$) will successfully distinguish between the two cases, solving the single-level EMostlyDISJ$_t(X)$ instance.
\end{proof}

A streaming algorithm $\mathcal{A}$ for $F_2$ naturally induces a communication protocol $\Pi$ for EMostlyDISJ$_t$, where player $i$ simulates $\mathcal{A}$ on their part of the stream and passes the memory state to player $i+1$. The total communication is $t \cdot \text{Space}(\mathcal{A})$.

\subsubsection{The Multi-Scale Construction and Direct Sum}
We now employ the multi-scale reduction framework of \cite{BravermanZ25} (Section 5.2). We define levels $l$, where $t_l = 2^l$, and conceptually embed an instance of EMostlyDISJ$_{t_l}$ at each level.

We define a range of ``good'' levels. 
A level $l$ is good if $t_l$ satisfies the necessary conditions for both the $F_2$ reduction and the communication lower bound:
\begin{enumerate}
    \item $t_l \leq \eps\sqrt{n}/2$ (Ensures the super-element size $d_l \geq 1$).
    \item $t_l \geq C \log(1/\delta)$ (Ensures the $F_2$ gap is distinguishable, Lemma \ref{lem:f2_reduction_mostly}).
    \item $\delta \ge t_l e^{-m_l/(6t_l)}$ (Condition for Corollary~\ref{cor:comm_cost_final}). Since $m_l/t_l = \Theta(t_l/\eps^2)$, this requires $t_l/\eps^2 \geq \Omega(\log(t_l/\delta))$.
\end{enumerate}

The number of good levels $L$ is approximately $\log_2(\eps\sqrt{n}) - \log_2(C' \max(1, \eps^2) \log(1/\delta))$. Assuming standard parameter regimes where $\eps$ is small, this yields $L = \Omega(\log(\frac{\eps\sqrt{n}}{\log(1/\delta)}))$.

A na\"{i}ve concatenation of these instances would result in a stream of length $\omega(n \log n)$, which would artificially inflate the space bounds. 
To embed these instances without increasing the total stream length beyond $n$, the construction interleaves them implicitly within a single continuous stream. 
The information cost is analyzed over a stream drawn from a distribution $\mathcal{E}$, where $n$ base elements are chosen uniformly and independently from $U$. 

To upper bound element repetitions across instances, we set the base universe size to $|U| = n^3$.
Let $\mathcal{E}_{coll}$ be the event that there are fewer than $\kappa = C_1\log\frac{1}{\delta}$ collisions in the entire stream of length $n$. 
Note that since the expected number of collisions is $\mu \le \binom{n}{2} / n^3 < \frac{1}{2n}$, then by standard tail bounds, the probability of $\kappa$ or more collisions is at most $O(\delta)$. 
Hence, for the remainder of the section, we condition on the event $\mathcal{E}_{coll}$. 

Now, because the stream consists of independent random variables, it can be analyzed at multiple scales simultaneously. 
For any level $l$, partitioning the stream into contiguous blocks of size $n/t_l$ to form $d_l$-tuples naturally generates a valid NO instance for the $t_l$-player EMostlyDISJ game.
However, before bounding the information measure, we must reconcile a subtle distributional difference between the multiscale stream and the communication game of Corollary~\ref{cor:comm_cost_final}. 
The communication inputs $X'$ are fixed-size sets drawn \textit{without replacement} from the universe $U^{d_l}$, whereas the continuous multiscale stream generates $d_l$-tuples \textit{with replacement}.


\begin{lemma}[Distributional Coupling]
\label{lem:coupling}
Let $\mathcal{D}_{\text{stream}}$ be the distribution of a single active stream block at level $l$, consisting of $m_l$ abstract $d_l$-tuples drawn i.i.d. with replacement from $U^{d_l}$. 
Let $\mathcal{D}_{\text{game}}$ be the exact communication distribution from Corollary~\ref{cor:comm_cost_final}, which generates sets \textit{without replacement}. 
For any streaming algorithm's memory state $M$ processing the stream $X \sim \mathcal{D}_{\text{stream}}$, the mutual information satisfies:
\[I(X ; M) \ge \left(1 - \frac{1}{n}\right) \cdot I(X_{\text{game}} ; M)\]
where $X_{\text{game}} \sim \mathcal{D}_{\text{game}}$.
\end{lemma}
\begin{proof}
Let $E$ be the event that the sequence of $m_l$ items drawn from $U^{d_l}$ contains no duplicate tuples within each player. 
Conditioned on $E$, a sequence drawn i.i.d.\ with replacement is identical to a uniform sequence drawn without replacement. 
Thus, $(X \mid E) \sim \mathcal{D}_{\text{game}}$.

Because we set the base universe size to $|U| \ge n^3$ and $d_l \ge 1$, the abstract universe size is $|U^{d_l}| = |U|^{d_l} \ge n^3$. 
The total number of abstract items drawn across all $t_l$ players in the instance is $m_l \le n$. 
Hence, the probability of drawing a duplicate is at most
\[\Pr[\neg E] \le \frac{m_l^2}{2|U^{d_l}|} \le \frac{n^2}{2n^3} \le \frac{1}{2n} \le \frac{1}{n}.\]

To lower bound the information cost, we use the chain rule for mutual information. 
Let $I_E$ be the binary indicator of event $E$. 
Expanding $I(X, I_E; M)$ in two ways, we have
\[I(X; M) + I(I_E; M \mid X) = I(I_E; M) + I(X; M \mid I_E).\]
Since $I_E$ is a deterministic function of $X$, we have $H(I_E \mid X) = 0$ and hence $I(I_E; M \mid X) = 0$. 
Since mutual information is non-negative, $I(I_E; M) \ge 0$. 
Thus $I(X; M) \ge I(X; M \mid I_E)$. 
Expanding the conditional mutual information over the states of $I_E$,
\begin{align*}
I(X; M \mid I_E) &= \Pr[E] \cdot I(X; M \mid E) + \Pr[\neg E] \cdot I(X; M \mid \neg E) \\
&\ge \Pr[E] \cdot I(X; M \mid E).
\end{align*}
Conditioned on $E$, the input is distributed as $\mathcal{D}_{\text{game}}$ and the memory state is produced by the same algorithm, so the joint law of $(X, M)$ given $E$ equals that of $(X_{\text{game}}, M)$; hence $I(X; M \mid E) = I(X_{\text{game}}; M)$. 
Substituting $\Pr[E] \ge 1 - 1/n$ gives the desired claim.

\end{proof}

We use the information measures defined by \cite{BravermanZ25}. Let $X$ be the random stream and $M_j$ the memory state of the algorithm after processing index $j$.

\begin{definition}[Definition 5.6 of \cite{BravermanZ25}]
For level $l$, with $t_l = 2^l$, define the information measure:
$$ I_l := \sum_{j=1}^{n} I\left(X_{(j - \frac{n}{2^l}, j - \frac{n}{2^{l+1}})}; M_j \mid M_{j - \frac{n}{2^l}}\right) $$
\end{definition}
This measure sums over all time steps $j$ to capture the total information the algorithm must retain about the past stream segments relevant to scale $l$.

\begin{lemma}[Adapted from Lemma 5.7 of \cite{BravermanZ25}]
\label{lem:I_l_bound_formal}
Let $S$ be the worst-case space (in bits) of the streaming algorithm. 
For each good level $l$, the information measure $I_l$ satisfies:
\[ I_l \ge \Omega\left(\frac{n}{\eps^2}\log\frac{1}{\delta}\right) - O(n \beta_l S) - O\left(\frac{n}{t_l} h_2(\beta_l)\right). \]
\end{lemma}
\begin{proof}
Following the structure of Lemma 5.7 in \cite{BravermanZ25}, we lower bound $I_l$ by tracking the number of independent instances of EMostlyDISJ$_{t_l}$ that can be embedded within the stream.
For a fixed level $l$, simulating a single instance of the game requires $t_l$ players sequentially passing the memory state of the streaming algorithm, corresponding to $t_l$ stream blocks of length $n/t_l$. 
Because the streaming algorithm uses at most $S$ bits of space, the induced communication protocol has a worst-case transcript length $C_{\Pi_l} \le t_l S$.

Evaluating this instance requires examining the sequence of memory states $M_j$ spaced at intervals of exactly $n/t_l$. 
However, this sequence of memory states need not begin at index $1$. 
Because the stream elements are drawn independently, shifting the starting index of these intervals yields disjoint set of indices $j$ defining $I_l$, and thus independent communication games. 
Since this starting offset can be shifted exactly $n/t_l$ times before the blocks begin to overlap, \cite{BravermanZ25} identified $\Omega(n/t_l)$ disjoint sequences of indices $j_1, \ldots, j_{t_l}$ that serve as independent parallel transcripts for the EMostlyDISJ$_{t_l}$ instances at that level.

By Corollary \ref{cor:comm_cost_final} (and noting $1/(1-\beta_l) \le O(1)$ since $\beta_l \le \delta \le 1/2$), the conditional information cost to solve one such instance is $\Omega(\frac{m_l}{t_l} \log(1/\delta)) - O(\beta_l C_{\Pi_l}) - O(h_2(\beta_l))$. 
Substituting $m_l = \Theta(t_l^2/\eps^2)$ from our parameter choice and $C_{\Pi_l} \le t_l S$, the cost per instance becomes $\Omega(\frac{t_l}{\eps^2} \log(1/\delta)) - O(\beta_l t_l S) - O(h_2(\beta_l))$.

To lower bound the information cost of a single embedded instance, we apply the conditioning argument required by Lemma \ref{lem:coupling}. 
Let $X \sim \mathcal{D}_{\text{stream}}$ be the stream instance. 
The coupling guarantees:
\[I(X ; \Pi) \ge \left(1 - \frac{1}{n}\right) \cdot I(X_{\text{game}} ; \Pi).\]
Because the streaming algorithm guarantees a failure probability of at most $\delta$ on \textit{any} valid input sequence, it fails on the communication game on the conditional distribution $X_{\text{game}}$ with probability less than $\delta$. 
Thus, Corollary~\ref{cor:comm_cost_final} applies directly to $I(X_{\text{game}} ; \Pi)$:
\[I(X_{\text{game}} ; \Pi) \ge \Omega\left(\frac{m_l}{t_l}\log\frac{1}{\delta}\right) - O(\beta_l C_{\Pi_l}) - O(h_2(\beta_l)).\]
Substituting $m_l = \Theta(t_l^2/\eps^2)$ and $C_{\Pi_l} \le t_l S$, the expected information cost per instance under $\mathcal{D}_{\text{stream}}$ is lower bounded by:
\[\left(1 - \frac{1}{n}\right) \left[ \Omega\left(\frac{t_l}{\eps^2}\log\frac{1}{\delta}\right) - O(\beta_l t_l S) - O(h_2(\beta_l)) \right].\]

Now, using Lemma 3.2 of \cite{BravermanZ25} (a chain rule for streaming mutual information), the total information cost $I_l$ is lower bounded by the sum of the costs of these independent interleaved instances. 
Therefore:
\begin{align*}
I_l &\geq (\text{Number of independent instances}) \times (\text{Information cost per instance}) \\
&\geq \Omega\left(\frac{n}{t_l}\right) \times \left[ \Omega\left(\frac{t_l}{\eps^2} \log \frac{1}{\delta}\right) - O(\beta_l t_l S) - O(h_2(\beta_l)) \right] \\
&= \Omega\left(\frac{n}{\eps^2} \log \frac{1}{\delta}\right) - O(n \beta_l S) - O\left(\frac{n}{t_l} h_2(\beta_l)\right).
\end{align*}
Observe that the dependence on $t_l$ cancels out in the main term. 
This gives a uniform lower bound on the required information cost across all good levels, independent of the number of players involved at scale $l$.
\end{proof}

The final step is the innovative direct sum argument over dependent instances developed by \cite{BravermanZ25}. They define the average information stored by the algorithm:

\begin{definition}[Definition 5.9 of \cite{BravermanZ25}]
$\bar{I} := \frac{1}{n}\sum_{j=1}^{n} I(X_{<j};M_j)$.
\end{definition}

They prove that the information measures $I_\ell$ from different scales essentially add up, even though the underlying problems are dependent.

\begin{lemma}[Lemma 5.10 of \cite{BravermanZ25}]
\label{lem:sum_info}
$\bar{I} \ge \frac{1}{n} \sum_{\ell \in \text{Good}} I_\ell$.
\end{lemma}
\begin{proof}
This follows from Lemma 5.8 of \cite{BravermanZ25}, which uses the chain rule to show that for any $j$, $I(X_{<j}; M_j)$ is lower bounded by the sum over $\ell$ of the information $M_j$ has about the specific past segment relevant to level $\ell$, conditioned on the memory state corresponding to the beginning of that segment. This relies on the observation that these past segments are disjoint for different $\ell$. Summing over $j$ and reordering yields the result.
\end{proof}

\begin{proof}[Proof of Theorem \ref{thm:main}]
Let $S$ be the space required by the algorithm. 
This space is at least the maximum information stored, which is at least the average information $\bar{I}$ (Observation 3.1 in \cite{BravermanZ25}), so $S \ge \bar{I}$.
By combining Lemma \ref{lem:I_l_bound_formal} and Lemma \ref{lem:sum_info}:
\begin{align*}
S \ge \bar{I} &\ge \frac{1}{n} \sum_{\ell \in \text{Good}} I_\ell \\
&\ge \frac{1}{n} \sum_{\ell \in \text{Good}} \left[ \Omega\left(\frac{n}{\eps^2}\log\left(\frac{1}{\delta}\right)\right) - O(n \beta_\ell S) - O\left(\frac{n}{t_\ell} h_2(\beta_\ell)\right) \right] \\
&= (\text{Number of Good Levels}) \cdot \Omega\left(\frac{1}{\eps^2}\log\left(\frac{1}{\delta}\right)\right) - S \sum_{\ell \in \text{Good}} O(\beta_\ell) - \sum_{\ell \in \text{Good}} O\left(\frac{h_2(\beta_\ell)}{t_\ell}\right).
\end{align*}
The number of good levels is $L = \Omega(\log(\frac{\eps\sqrt{n}}{\log(1/\delta)}))$. 
Observe that the abort probabilities $\beta_\ell \le t_\ell e^{-m_\ell/(6t_\ell)}$ decay exponentially across the dyadic levels since $m_\ell / t_\ell = \Theta(t_\ell/\eps^2)$. 
Thus, by choosing constants appropriately, we ensure that the aggregated defect sum is bounded by a fraction: $\sum_{\ell \in \text{Good}} O(\beta_\ell) \le 1/2$. 
The entropy penalty sum is similarly upper bounded by a negligible constant $O(1)$. 

Substituting these bounds into our inequality yields:
\[S \ge \Omega\left(\log\left(\frac{\eps\sqrt{n}}{\log\frac{1}{\delta}}\right)\frac{1}{\eps^2} \log\left(\frac{1}{\delta}\right)\right) - \frac{1}{2}S - O(1).\]
Thus, we conclude:
\[S \ge \Omega\left(\log\left(\frac{\eps\sqrt{n}}{\log\frac{1}{\delta}}\right)\frac{1}{\eps^2} \log\left(\frac{1}{\delta}\right)\right).\]
\end{proof}

\section{Improved Algorithms}

\subsection{Bounded Frequencies}

We present an $\ell_2$-estimation algorithm for insertion-only streams where the frequencies are bounded by $B$. 
The goal is to achieve a space complexity where the dependence on $B$ is only logarithmic, and the dependence on $n$ is also logarithmic, independent of $1/\eps^2$.

\begin{theorem}\label{thm:l2_bounded_freq_improved}
Let $x \in \mathbb{Z}_{\geq 0}^n$ be the frequency vector of an insertion-only stream of length $m$. 
Suppose that $0 \leq x_i \leq B$ for all $i \in [n]$. 
Then there exists a streaming algorithm that, with probability at least $1 - \delta$, computes a $(1 \pm \eps)$-approximation to $\|x\|_2^2$ using space (in bits)
\[O\left(\frac{1}{\eps^2}\log^2\frac{B}{\eps}\log\frac{1}{\delta}\left(\log B + \log\frac{1}{\eps}\right)\right) + O\left(\log\frac{n}{\delta}\right).\]
\end{theorem}
\begin{proof}
The strategy is to use subsampling to reduce the effective dimension of the problem to a size that depends polynomially on $B$ and $\frac{1}{\eps}$, but only logarithmically on $n$. 
Then, we apply an efficient $\ell_2$ estimator on this reduced stream.

We use pairwise independent hash functions $h : [n] \to [2^L]$ (where $L = O(\log n)$). 
We define levels of sampling based on the number of trailing zeros in the hash value $h(i)$. 
Let $z(i)$ be the number of trailing zeros of $h(i)$. 
We aim to find a sampling level $\ell$ such that the expected number of distinct items $i$ with $z(i) \geq \ell$ is approximately $K$, where $K = \Theta\left(\frac{B^2}{\eps^2}\right)$, independent of $\delta$. 

We first utilize a continuous $F_0$ tracking algorithm \cite{kane2010optimal} to estimate the total number of distinct elements $F_0$. 
By setting the failure probability of the $F_0$ estimator to $\delta/\log n$, we may assume that it yields a constant-factor approximation to $F_0$ at all times throughout the stream. 
Let $F_{0,t}$ be the number of non-zero entries at time $t$ of the stream, and let $\hat{F}_{0,t}$ be the corresponding estimator. 
We first describe a two-pass version of our algorithm, and then explain how to extend it to a single pass.
    
\paragraph{A two-pass algorithm.} 
On the first pass, we compute the final estimate $\hat{F}_0$. 
We choose the sampling level $\ell$ such that $2^\ell \approx \hat{F}_0 / K$. 
This ensures that the expected number of items surviving the sampling is within a constant factor of $K$. 
Let $S$ be the set of surviving coordinates: $S = \{i \in [n] \mid z(i) \geq \ell\}$. 
The expected size of $S$ is roughly $F_0 / 2^\ell \approx K$.

We define the natural inverse-probability estimator based on the sampled set $S$: 
\[Y = 2^\ell \sum_{i \in S} x_i^2 = \frac{1}{p} \sum_{i \in S} x_i^2,\]
where $p = 2^{-\ell}$. 
Since $Y$ is a standard sampling estimator and elements are sampled pairwise independently, its expectation is $\mathbb{E}[Y] = \|x\|_2^2$, and its variance is upper bounded by:
\[\text{Var}(Y) \leq \frac{1-p}{p} \sum_{i=1}^n x_i^4 \leq \frac{1}{p} \|x\|_4^4.\]
Since the frequencies are upper bounded by $B$ ($x_i \leq B$), we have $\|x\|_4^4 \leq B^2 \|x\|_2^2$, which implies $\text{Var}(Y) \leq \frac{B^2}{p} \|x\|_2^2$.

To achieve our $(1 \pm \eps)$ guarantee, Chebyshev's inequality implies that it suffices to have $\text{Var}(Y) \leq \eps^2 \|x\|_2^4$. 
This holds if $p \geq \Omega\left(\frac{B^2}{\eps^2 \|x\|_2^2}\right)$. 
Since all non-zero entries are at least $1$, $\|x\|_2^2 \geq F_0$. 
Therefore, it suffices to set $p \geq \Omega\left(\frac{B^2}{\eps^2 F_0}\right)$. 
Our choice of $\ell$, which ensures $p \approx K / \hat{F}_0$, guarantees this condition when the constants are set appropriately. 
Thus, a single sampling instance gives a $(1 \pm \eps)$-approximation with constant probability, e.g., $7/8$. 
Finally, to achieve the required $1-\delta$ success probability, confidence bound, we utilize the standard median-of-means technique, taking the median of $T = \Theta(\log(1/\delta))$ independent instances. 

\paragraph{Implementation and space complexity.} 
We do not explicitly store the sampled set $S_r$ for each independent instance $r$. 
Instead, for each of the $T$ independent instances, we observe that the stream restricted to the coordinates in $S_r$ is itself a valid sparse stream. 
To avoid paying a logarithmic dependence on $n$ for the sketch hash seeds, we apply dimensionality reduction. 
Let $M = \Theta(K^2)$. 
We draw a pairwise independent hash function $g_r : [n] \to [M]$. 
Because the expected support size of the restricted stream is $K$, hashing into $M = \Theta(K^2)$ buckets yields zero collisions with high constant probability. 
Let $y^{(r)}$ be this compressed vector. 
When a stream update $(i, \Delta)$ arrives, we evaluate $h_r(i)$ to check if $z_r(i) \ge \ell$. 
If so, we pass the update to a standard AMS sketch \cite{AlonMS99} tracking $x|_{S_r}$.

Because $K = \Theta(B^2/\eps^2)$ is independent of $\delta$, the expected $\ell_1$ norm of each restricted substream $x|_S$ is at most $BK = \Theta(B^3/\eps^2)$. 
By standard concentration bounds, the maximum $\ell_1$ norm is $O(BK)$ with high probability. 
Thus, the bit complexity for a single counter is at most $\log(BK) = O(\log B + \log(1/\eps))$. 

Each of the $T = \Theta(\log(1/\delta))$ independent instances maintains an AMS sketch with $O(1/\eps^2)$ counters, operating on the reduced domain $[M]$. 
Thus, each of the $O(1/\eps^2)$ hash functions uses $O(\log M) = O(\log K) = O(\log B + \log(1/\eps))$ bits to store. 
Hence, the overall space required for the AMS sketches is:
\[O\left( \frac{1}{\eps^2} \log\left(\frac{1}{\delta}\right) \left(\log B + \log\frac{1}{\eps}\right) \right).\]
We also need a constant-factor approximation to determine our sampling probability, through a continuous monitoring $F_0$ algorithm that uses an additive $O\left(\log\frac{n}{\delta}\right)$ bits, c.f., Theorem~\ref{thm:fzero:monitor}. 

\paragraph{Extension to a single pass.} 
To extend this to a single pass, one could run $O(\log n)$ instances of the above two-pass algorithm in parallel, using every possible power-of-two sampling rate $p$. 
At the end of the stream, we would query $\hat{F}_0$ and output the estimate from the instance that used the correct value of $p$. 
However, this would increase the overall space by an undesired $\log n$ factor. 
    
To avoid this overhead, we utilize our continuous $F_0$ tracker. 
Let $t_0 < t_1 < \ldots$ be the timestamps at which the $F_0$ estimate crosses consecutive powers of two, meaning $\hat{F}_{0,t_i} \geq 2^i$. 
Let $N= \lceil \log(B^2/\eps^2) \rceil$. 
At each milestone $t_i$, we initialize a new group of $N+1$ instances of our sketch, which consists of the $T = \Theta(\log(1/\delta))$ independent instances necessary for the median-of-means aggregation.  
These $N+1$ instances use sampling probabilities $p$ calibrated for future $F_0$ scales of $2^0 \hat{F}_{0,t_i}, 2^1 \hat{F}_{0,t_i}, \ldots, 2^{N} \hat{F}_{0,t_i}$. 

Crucially, we do not store every sketch indefinitely. 
We maintain only the sketch groups generated at the $N+1$ most recent milestones (from $t_i$ down to $t_{i-N}$). 
Any older sketches are permanently discarded. 
Consequently, we maintain at most $O(N^2) = O(\log^2(B/\eps))$ active sketches at any point in time.

At the end of the stream, we examine the earliest remaining group of non-discarded sketches. 
Among the $N$ sketches in this group, we select the one that used the correct sampling rate $p$ corresponding to our final $F_0$ estimate. 
This selected sketch guarantees a $(1 \pm \eps)$-approximation to $F_2$, but only for the \emph{suffix} of the stream processed after the sketch was initialized. 
    
Let $x_2$ be the frequency vector corresponding to this suffix, and let $x_1$ be the discarded prefix, so that $x = x_1 + x_2$. Because we retained sketches stretching back $N = \lceil\log(B^2/\eps^2)\rceil$ epochs, the support size of the unsketched prefix $x_1$ satisfies:
\[\|x_1\|_0 \leq 2^{-N}\|x\|_0 \leq \frac{\eps^2}{B^2}\|x\|_0.\]
Exploiting the bounded-frequency assumption ($x_i \leq B$) and the fact that all non-zero entries are at least $1$ (so $\|x\|_0 \leq \|x\|_2^2$), we can bound the $\ell_2$ norm of the prefix:
\[\|x_1\|_2^2 \leq B^2 \|x_1\|_0 \leq B^2 \left(\frac{\eps^2}{B^2}\|x\|_0\right) = \eps^2 \|x\|_0 \leq \eps^2 \|x\|_2^2.\]
Taking the square root yields $\|x_1\|_2 \leq \eps \|x\|_2$. 
By the triangle inequality, the difference between the full norm $\|x\|_2$ and the estimated suffix norm $\|x_2\|_2$ is at most $\|x_1\|_2 \leq \eps \|x\|_2$. 
Adjusting $\eps$ by a constant factor ensures a strict $(1 \pm \eps)$-approximation for the final estimate. 
    
The total space complexity is dominated by the $F_0$ tracker plus the $O(\log^2(B/\eps))$ active epoch states, where each state maintains $T = \Theta(\log(1/\delta))$ independent AMS instances. 
Multiplying the space complexity for each instance by these factors gives the final desired bound. 
\end{proof}

\subsection{Sparse Streams}

We present an $\ell_2$-estimation algorithm for insertion-only streams where the frequency vector $x$ is $k$-sparse, meaning it has at most $k$ non-zero entries (i.e., $\|x\|_0 \leq k$). This algorithm achieves a space complexity where $k$ and the stream length $m$ appear only inside logarithmic factors. The core idea is to first use hashing to reduce the dimensionality of $x$. We then apply the standard AMS sketch to the resulting compressed vector. To aggressively optimize the space required per counter, we do not maintain the exact sketch counters. Instead, we conceptually split each counter into positive and negative increments and maintain them approximately using Morris counters.

\begin{theorem}
\label{thm:l2_sparse_improved}
Let $x \in \mathbb{Z}_{\geq 0}^n$ be the frequency vector of an insertion-only stream of length $m$. Suppose that $\|x\|_0 \leq k$. 
Then there exists a streaming algorithm that, with probability at least $1 - \delta$, computes a $(1 \pm \eps)$-approximation to $\|x\|_2^2$ using space (in bits)
\[O\left(\frac{1}{\eps^2}\log\left(\frac{1}{\delta}\right)\left(\log\frac{k}{\eps} + \log\log m\right) + \log n \cdot \log\frac{1}{\delta} \right).\]
\end{theorem}
\begin{proof}
The strategy involves a two-stage process: dimensionality reduction tailored specifically for sparse vectors, followed by an AMS sketch where the underlying counters are maintained approximately using Morris counters \cite{morris1978counting}.

\paragraph{Dimensionality reduction via hashing.} 
We first reduce the dimension by hashing the coordinates of $x$ into a much smaller set of random buckets. 
We draw $T = \Theta(\log(1/\delta))$ independent pairwise independent hash functions $h_r : [n] \to [M]$, where $M = \lceil C k / \eps^2 \rceil$ for a sufficiently large constant $C$. 
For each independent instance $r\in[T]$, we define the reduced vector $y^{(r)} \in \mathbb{R}^M$, where $y^{(r)}_j = \sum_{i : h_r(i) = j} x_i$. 
Since $x$ is $k$-sparse, the expected $\ell_2^2$ error introduced by collisions in a single instance is at most $\|x\|_1^2 / M \le k \|x\|_2^2 / M$. 
By Markov's inequality, a single hash preserves $\|y^{(r)}\|_2^2 = (1 \pm \eps/4)\|x\|_2^2$ with constant probability. 
Taking the median across the $T$ independent instances amplifies this probability to the desired $1-\delta$.

\paragraph{AMS Sketch with Morris Counters.} 
For each of the $T$ instances, we apply an AMS sketch to the lower-dimensional vector $y^{(r)}$ using $R = \Theta(1/\eps^2)$ estimators. 
We conceptually decompose each exact AMS counter into its positive and negative accumulating parts, i.e., $A_{r,w} = A_{r,w}^+ - A_{r,w}^-$. 
Notice that the total accumulated mass across both parts is precisely the $\ell_1$ norm of $x$: $A_{r,w}^+ + A_{r,w}^- = \|x\|_1 = m$.
To avoid the space overhead of storing these counters exactly, we maintain them approximately. 
We instantiate two unbiased Generalized Morris counters for each index $(r, w)$, one for $A_{r,w}^+$ and one for $A_{r,w}^-$, and each with variance $\gamma = \Theta(\eps^2 / k^2)$. 
Let $\hat{A}_{r,w}^+$ and $\hat{A}_{r,w}^-$ be their respective estimates. 
Our approximate counter is then $\hat{A}_{r,w} = \hat{A}_{r,w}^+ - \hat{A}_{r,w}^-$, which remains unbiased, i.e., $\mathbb{E}[\hat{A}_{r,w}] = A_{r,w}$. 
Thus, the variance of the approximate counter is at most
\[\text{Var}(\hat{A}_{r,w}) = \text{Var}(\hat{A}_{r,w}^+) + \text{Var}(\hat{A}_{r,w}^-) \le \gamma (A_{r,w}^+)^2 + \gamma (A_{r,w}^-)^2 \le \gamma m^2.\]

\paragraph{Error analysis.} 
Let $\hat{Y}'_r$ be the AMS estimator computed using these approximate counters for a single instance $r$, i.e., averaging $R = \Theta(1/\eps^2)$ squared approximate counters $\hat{A}_{r,w}^2$ across $w\in[R]$. 
Because the Morris counters are unbiased, the expected value of a squared approximate counter $\hat{A}_{r,w}^2$ includes the true squared signal plus the variance, i.e., $\mathbb{E}[\hat{A}_{r,w}^2] = A_{r,w}^2 + \text{Var}(\hat{A}_{r,w})$. 
Thus, the expected bias introduced to the final estimator by the approximate counters is at most $\text{Var}(\hat{A}_{r,w}) \le \gamma m^2\le O(\eps^2 m^2/k^2)$. 
Because $x$ is $k$-sparse and $\|x\|_1=m$, then $\|x\|_2^2 \ge m^2/k$, so that the bias is at most $O(\eps^2/k)\cdot\|x\|_2^2$. 

Similarly, we bound the variance of the overall estimator. 
Let $E = \hat{A}_{r,w} - A_{r,w}$ be the error of the Morris counter. 
Expanding the squared counter, we have $\hat{A}_{r,w}^2 = (A_{r,w} + E)^2 = A_{r,w}^2 + 2A_{r,w}E + E^2$.  
Since $A_{r,w}^2=O(m^2)$ and $\text{Var}(E) \le \gamma m^2$, then the variance is at most $O(\gamma m^4)$. 
Averaging over the $R = \Theta(1/\eps^2)$ and applying the above argument, it follows that the variance is at most $O(\eps^4 \|x\|_2^4) \ll \eps^2 \|x\|_2^4$. 
Thus Chebyshev's inequality, each instance outputs a $(1 \pm \eps/2)$ approximation to $\|y^{(r)}\|_2^2$ with constant probability. 
By taking the median of $O(\log 1/\delta)$ instances, we can boost the success probability to $1-\delta$. 

\paragraph{Space Complexity.} 
A Generalized Morris counter tracking up to $m$ with variance parameter $\gamma = \Theta(\eps^2/k^2)$ uses $O(\log(1/\gamma) + \log\log m) = O(\log(k/\eps) + \log\log m)$ bits. 
Maintaining $2RT = O(\frac{1}{\eps^2}\log\frac{1}{\delta})$ such counters thus uses space
\[O\left(\frac{1}{\eps^2}\log\left(\frac{1}{\delta}\right) \left(\log\frac{k}{\eps} + \log\log m\right)\right).\]
Additionally, the $T$ initial pairwise independent hash functions $h_r : [n] \to [M]$ require $O(\log(1/\delta) \log n)$ additional bits of space. 
Similarly, the AMS hash functions $\sigma$ require $O(\log M) = O(\log(k/\eps))$ bits to generate. 
Altogether, the space complexity is
\[O\left(\frac{1}{\eps^2}\log\left(\frac{1}{\delta}\right)\left(\log\frac{k}{\eps} + \log\log m\right) + \log n \cdot \log\frac{1}{\delta} \right).\]
\end{proof}

\bibliographystyle{alpha}
\bibliography{references}

\newcommand{\etalchar}[1]{$^{#1}$}
\begin{thebibliography}{KNPW11}

\bibitem[AGMS02]{AlonGMS02}
Noga Alon, Phillip~B. Gibbons, Yossi Matias, and Mario Szegedy.
\newblock Tracking join and self-join sizes in limited storage.
\newblock {\em J. Comput. Syst. Sci.}, 64(3):719--747, 2002.

\bibitem[AMS99]{AlonMS99}
Noga Alon, Yossi Matias, and Mario Szegedy.
\newblock The space complexity of approximating the frequency moments.
\newblock {\em J. Comput. Syst. Sci.}, 58(1):137--147, 1999.

\bibitem[BDN17]{BlasiokDN17}
Jaroslaw Blasiok, Jian Ding, and Jelani Nelson.
\newblock Continuous monitoring of $\ell_p$ norms in data streams.
\newblock In {\em Approximation, Randomization, and Combinatorial Optimization. Algorithms and Techniques, {APPROX/RANDOM}}, volume~81, pages 32:1--32:13, 2017.

\bibitem[BEO22]{Ben-EliezerEO22}
Omri Ben{-}Eliezer, Talya Eden, and Krzysztof Onak.
\newblock Adversarially robust streaming via dense-sparse trade-offs.
\newblock In {\em 5th Symposium on Simplicity in Algorithms, SOSA@SODA}, pages 214--227, 2022.

\bibitem[BJKS04]{Bar-YossefJKS04}
Ziv Bar{-}Yossef, T.~S. Jayram, Ravi Kumar, and D.~Sivakumar.
\newblock An information statistics approach to data stream and communication complexity.
\newblock {\em J. Comput. Syst. Sci.}, 68(4):702--732, 2004.

\bibitem[BKSV14]{BravermanKSV14}
Vladimir Braverman, Jonathan Katzman, Charles Seidell, and Gregory Vorsanger.
\newblock An optimal algorithm for large frequency moments using {O}$(n^{1-2/k})$ bits.
\newblock In {\em Approximation, Randomization, and Combinatorial Optimization. Algorithms and Techniques, {APPROX/RANDOM}}, pages 531--544, 2014.

\bibitem[Bla20]{Blasiok20}
Jaroslaw Blasiok.
\newblock Optimal streaming and tracking distinct elements with high probability.
\newblock {\em {ACM} Trans. Algorithms}, 16(1):3:1--3:28, 2020.

\bibitem[BO13]{BravermanO13}
Vladimir Braverman and Rafail Ostrovsky.
\newblock Approximating large frequency moments with pick-and-drop sampling.
\newblock In {\em Approximation, Randomization, and Combinatorial Optimization. Algorithms and Techniques - 16th International Workshop, {APPROX}, and 17th International Workshop, {RANDOM}. Proceedings}, 2013.

\bibitem[BVWY18]{BravermanVWY18}
Vladimir Braverman, Emanuele Viola, David~P. Woodruff, and Lin~F. Yang.
\newblock Revisiting frequency moment estimation in random order streams.
\newblock In {\em 45th International Colloquium on Automata, Languages, and Programming, {ICALP}}, pages 25:1--25:14, 2018.

\bibitem[BZ25]{BravermanZ25}
Mark Braverman and Or~Zamir.
\newblock Optimality of frequency moment estimation.
\newblock In Michal Kouck{\'{y}} and Nikhil Bansal, editors, {\em Proceedings of the 57th Annual {ACM} Symposium on Theory of Computing, {STOC}}, pages 360--370, 2025.

\bibitem[CKS03]{ChakrabartiKS03}
Amit Chakrabarti, Subhash Khot, and Xiaodong Sun.
\newblock Near-optimal lower bounds on the multi-party communication complexity of set disjointness.
\newblock In {\em 18th Annual {IEEE} Conference on Computational Complexity}, pages 107--117, 2003.

\bibitem[Gan11]{Ganguly11}
Sumit Ganguly.
\newblock Polynomial estimators for high frequency moments.
\newblock {\em CoRR}, abs/1104.4552, 2011.

\bibitem[GGI{\etalchar{+}}02]{GilbertGIKMS02}
Anna~C. Gilbert, Sudipto Guha, Piotr Indyk, Yannis Kotidis, S.~Muthukrishnan, and Martin Strauss.
\newblock Fast, small-space algorithms for approximate histogram maintenance.
\newblock In {\em Proceedings on 34th Annual {ACM} Symposium on Theory of Computing}, pages 389--398, 2002.

\bibitem[GLW{\etalchar{+}}25]{GribelyukLWYZ25}
Elena Gribelyuk, Honghao Lin, David~P. Woodruff, Huacheng Yu, and Samson Zhou.
\newblock Lifting linear sketches: Optimal bounds and adversarial robustness.
\newblock In {\em Proceedings of the 57th Annual {ACM} Symposium on Theory of Computing, {STOC}}, pages 395--406, 2025.

\bibitem[GLW{\etalchar{+}}26]{GribelyukLWYZ26}
Elena Gribelyuk, Honghao Lin, David~P. Woodruff, Huacheng Yu, and Samson Zhou.
\newblock Adversarial robustness on insertion-deletion streams.
\newblock In {\em Proceedings of the 58th Annual {ACM} Symposium on Theory of Computing, {STOC}}, pages 2278--2289, 2026.

\bibitem[GW18]{GangulyW18}
Sumit Ganguly and David~P. Woodruff.
\newblock High probability frequency moment sketches.
\newblock In {\em 45th International Colloquium on Automata, Languages, and Programming, {ICALP}}, 2018.

\bibitem[Ind06]{Indyk06}
Piotr Indyk.
\newblock Stable distributions, pseudorandom generators, embeddings, and data stream computation.
\newblock {\em J. {ACM}}, 53(3):307--323, 2006.

\bibitem[IW05]{IndykW05}
Piotr Indyk and David~P. Woodruff.
\newblock Optimal approximations of the frequency moments of data streams.
\newblock In {\em Proceedings of the 37th Annual {ACM} Symposium on Theory of Computing (STOC)}, pages 202--208, 2005.

\bibitem[JWZ24]{JayaramWZ24}
Rajesh Jayaram, David~P. Woodruff, and Samson Zhou.
\newblock Streaming algorithms with few state changes.
\newblock {\em Proc. {ACM} Manag. Data}, 2(2):82, 2024.

\bibitem[KNPW11]{KaneNPW11}
Daniel~M. Kane, Jelani Nelson, Ely Porat, and David~P. Woodruff.
\newblock Fast moment estimation in data streams in optimal space.
\newblock In {\em Proceedings of the 43rd {ACM} Symposium on Theory of Computing, {STOC}}, pages 745--754, 2011.

\bibitem[KNW10a]{KaneNW10}
Daniel~M. Kane, Jelani Nelson, and David~P. Woodruff.
\newblock On the exact space complexity of sketching and streaming small norms.
\newblock In {\em Proceedings of the Twenty-First Annual {ACM-SIAM} Symposium on Discrete Algorithms, {SODA}}, pages 1161--1178, 2010.

\bibitem[KNW10b]{kane2010optimal}
Daniel~M Kane, Jelani Nelson, and David~P Woodruff.
\newblock An optimal algorithm for the distinct elements problem.
\newblock In {\em Proceedings of the twenty-ninth ACM SIGMOD-SIGACT-SIGART symposium on Principles of database systems}, pages 41--52, 2010.

\bibitem[KPW21]{KPW21}
Akshay Kamath, Eric Price, and David~P. Woodruff.
\newblock A simple proof of a new set disjointness with applications to data streams.
\newblock In Valentine Kabanets, editor, {\em 36th Computational Complexity Conference, {CCC} 2021, July 20-23, 2021, Toronto, Ontario, Canada (Virtual Conference)}, volume 200 of {\em LIPIcs}, pages 37:1--37:24. Schloss Dagstuhl - Leibniz-Zentrum f{\"{u}}r Informatik, 2021.

\bibitem[KSZC03]{KrishnamurthySZC03}
Balachander Krishnamurthy, Subhabrata Sen, Yin Zhang, and Yan Chen.
\newblock Sketch-based change detection: methods, evaluation, and applications.
\newblock In {\em Proceedings of the 3rd {ACM} {SIGCOMM} Internet Measurement Conference, {IMC}}, pages 234--247, 2003.

\bibitem[Li08]{Li08}
Ping Li.
\newblock Estimators and tail bounds for dimension reduction in $\ell_{\alpha} (0<\alpha\le 2)$ using stable random projections.
\newblock In {\em Proceedings of the Nineteenth Annual {ACM-SIAM} Symposium on Discrete Algorithms, {SODA}}, pages 10--19, 2008.

\bibitem[LSW{\etalchar{+}}26]{Lin0W0Z26}
Honghao Lin, Zhao Song, David~P. Woodruff, Shenghao Xie, and Samson Zhou.
\newblock $l_p$ sampling in distributed data streams with applications to adversarial robustness.
\newblock In {\em Proceedings of the 2026 Annual {ACM-SIAM} Symposium on Discrete Algorithms, {SODA}}, pages 4342--4409, 2026.

\bibitem[Mor78]{morris1978counting}
Robert Morris.
\newblock Counting large numbers of events in small registers.
\newblock {\em Communications of the ACM}, 21(10):840--842, 1978.

\bibitem[Mut05]{muthukrishnan2005data}
Shanmugavelayutham Muthukrishnan.
\newblock Data streams: Algorithms and applications.
\newblock {\em Foundations and Trends{\textregistered} in Theoretical Computer Science}, 1(2):117--236, 2005.

\bibitem[Woo04]{Woodruff04}
David~P. Woodruff.
\newblock Optimal space lower bounds for all frequency moments.
\newblock In {\em Proceedings of the Fifteenth Annual {ACM-SIAM} Symposium on Discrete Algorithms, {SODA}}, pages 167--175, 2004.

\bibitem[Woo14]{woodruff2014sketching}
David~P Woodruff.
\newblock Sketching as a tool for numerical linear algebra.
\newblock {\em Foundations and Trends{\textregistered} in Theoretical Computer Science}, 10(1--2):1--157, 2014.

\bibitem[WZ21a]{WoodruffZ21b}
David~P. Woodruff and Samson Zhou.
\newblock Separations for estimating large frequency moments on data streams.
\newblock In {\em 48th International Colloquium on Automata, Languages, and Programming, {ICALP}}, pages 112:1--112:21, 2021.

\bibitem[WZ21b]{WoodruffZ21}
David~P. Woodruff and Samson Zhou.
\newblock Tight bounds for adversarially robust streams and sliding windows via difference estimators.
\newblock In {\em 62nd {IEEE} Annual Symposium on Foundations of Computer Science, {FOCS}}, pages 1183--1196, 2021.

\bibitem[WZ24]{WoodruffZ24}
David~P. Woodruff and Samson Zhou.
\newblock Adversarially robust dense-sparse tradeoffs via heavy-hitters.
\newblock In {\em Advances in Neural Information Processing Systems 37: Annual Conference on Neural Information Processing Systems, NeurIPS}, 2024.

\end{thebibliography}

\end{document}